\documentclass[a4paper,british,cleveref,autoref,thm-restate,final]{lipics-v2021}
\pdfoutput=1

\usepackage{xparse}
\usepackage{etoolbox}
\usepackage{xspace}
\usepackage{environ}

\usepackage[utf8]{inputenc}

\usepackage{xcolor}

\usepackage{babel}

\usepackage{amsmath}
\usepackage{amsthm}
\usepackage{amsfonts}
\usepackage{amssymb}
\usepackage{stmaryrd}
\SetSymbolFont{stmry}{bold}{U}{stmry}{m}{n}
\usepackage{bm}
\usepackage{mathtools}
\usepackage{centernot}
\usepackage{scalerel}
\usepackage{lscape}

\allowdisplaybreaks

\makeatletter
\newcommand{\raisemath}[1]{\mathpalette{\raisem@th{#1}}}
\newcommand{\raisem@th}[3]{\raisebox{#1}{$#2#3$}}
\makeatother

\usepackage{array}
\usepackage{booktabs}

\usepackage{float}
\usepackage{placeins}

\usepackage[inline]{enumitem}

\usepackage{graphicx}
\usepackage{tikz}
\usetikzlibrary{arrows,positioning,arrows.meta,decorations.pathmorphing}
\tikzset{
  squiggly/.style = {
    line join=round,
    decorate, decoration={
      zigzag,
      segment length=4,
      amplitude=.9,post=lineto,
      post length=2pt}
  }
}

\makeatletter
\newcommand{\customlabel}[4][0]{%
	\protected@write\@auxout{}{\string\newlabel{#3}{{#4}{\thepage}{#4}{#3}{}}}%
	\protected@write\@auxout{}{\string\newlabel{#3@cref}{{[#2][#1][#1]#4}{\thepage}}}%
}

\newcommand{\crefv}[1]{%
	\begingroup\@cref@compressfalse\@cref@sortfalse\cref{#1}\endgroup%
}
\newcommand{\Crefv}[1]{%
	\begingroup\@cref@compressfalse\@cref@sortfalse\Cref{#1}\endgroup%
}
\newcommand{\crefabbrev}[1]{%
	\begingroup\@cref@abbrevtrue\cref{#1}\endgroup%
}
\makeatother

\AtEndPreamble{\hypersetup{
	final,
}}

\usepackage{breakurl}

\usepackage{proof-dashed}

\newcounter{rule}
\newcommand{\rname}[1]{\ensuremath{\textsc{#1}}}
\newcommand{\rlabel}[2]{{%
	\refstepcounter{rule}%
	\customlabel[\therule]{infrule}{#2}{#1}%
	\hypertarget{#2}{#1}}}
\crefname{infrule}{rule}{rules}
\Crefname{infrule}{Rule}{Rules}

\newcommand{\rtag}[1]{\text{\footnotesize[#1]}}

\def\rulesvsep{7pt}

\usepackage{listings}

\usepackage[
 	disable,     %
	obeyFinal,   %
	textsize=tiny,
	colorinlistoftodos,
	prependcaption
]{todonotes}

\newcommand{\chorlam}{\text{\upshape Chor$\lambda$}\xspace}

\let\eoe\qed

\DeclareMathOperator{\fv}{fv}
\DeclareMathOperator{\roles}{roles}
\DeclareMathOperator{\type}{type}
\DeclareMathOperator{\sroles}{sroles}
\DeclareMathOperator{\dist}{distinct}

\newcommand{\case}[5]{\ccase\;#1\;\cof\;\cleft{#2}{#3}\csep\cright{#4}{#5}}
\newcommand{\DN}{\mathbb{D}}
\newcommand{\DP}{\mathbb{D}}
\newcommand{\epp}[3][]{\ensuremath{\left\llbracket{#2}\right\rrbracket^{#1}_{#3}}}
\renewcommand{\merge}{\sqcup}

\newcommand{\funf}[1]{\ensuremath{\mathsf{#1}}}
\newcommand{\typef}[1]{\ensuremath{\mathsf{#1}}}
\newcommand{\unitt}{\typef{()}}
\newcommand{\pairt}[2]{\ensuremath{{#1}\times{#2}}}

\newcommand{\fundef}[1]{\funf{#1}}
\newcommand{\fun}[2]{\lambda {\var{#1}:#2}.}
\let\var\mathit

\newcommand{\keyword}[1]{\boldsymbol{\mathsf{#1}}}
\newcommand{\comsymbol}{\keyword{com}}
\newcommand{\com}[2]{\comsymbol_{#1, #2}}
\newcommand{\select}[2]{\keyword{select}_{#1, #2}}
\newcommand{\send}[1]{\keyword{send}_{#1}}
\newcommand{\recv}[1]{\keyword{recv}_{#1}}

\newcommand{\inl} {\keyword{inl}}
\newcommand{\inr} {\keyword{inr}}
\newcommand{\Inl} {\keyword{Inl}}
\newcommand{\Inr} {\keyword{Inr}}
\newcommand{\pair}{\keyword{pair}}
\newcommand{\Pair}{\keyword{Pair}}
\newcommand{\fst} {\keyword{fst}}
\newcommand{\snd} {\keyword{snd}}

\newcommand{\cif}  {\keyword{if}}
\newcommand{\cthen}{\keyword{then}}
\newcommand{\celse}{\keyword{else}}
\newcommand{\ccase}{\keyword{case}}
\newcommand{\cof}  {\keyword{of}}
\newcommand{\csep}{{\boldsymbol{;}\;}}
\newcommand{\cleft}[1]{\Inl \; #1 \Rightarrow}
\newcommand{\cright}[1]{\Inr \; #1 \Rightarrow}

\newcommand{\unit}{\keyword{()}}

\newcommand{\litAt}[2]{#1@{#2}}
\newcommand{\typeAt}[2]{#1@{#2}}
\newcommand{\typefAt}[2]{\typeAt{\typef{#1}}{#2}}
\newcommand{\lab}[1]{\typef{#1}}
\newcommand{\role}[1]{{#1}}
\newcommand{\rolef}[1]{\mathsf{#1}}

\newcommand{\botm}{\keyword{\bot}}
\newcommand{\botmt}{\typef{\bot}}

\newcommand{\Bool}{\typef{Bool}}
\newcommand{\Int}{\typef{Int}}
\newcommand{\String}{\typef{String}}
\newcommand{\List}[1]{{[#1]}}
\newcommand{\UnitAt}[1]{\typeAt{\unitt}{#1}}
\newcommand{\BoolAt}[1]{\typefAt{\Bool}{#1}}
\newcommand{\IntAt}[1]{\typefAt{\Int}{#1}}
\newcommand{\StringAt}[1]{\typefAt{\String}{#1}}
\newcommand{\ListAt}[2]{\typeAt{\List{#1}}{#2}}
\title{Functional Choreographic Programming}

\author{Luís Cruz-Filipe}{Department of Mathematics and Computer Science, University of Southern Denmark, Denmark}{lcfilipe@gmail.com}{https://orcid.org/0000-0002-7866-7484}{}

\author{Eva Graversen}{Department of Mathematics and Computer Science, University of Southern Denmark, Denmark}{efgraversen@imada.sdu.dk}{0000-0002-9430-4907}{}

\author{Lovro Lugović}{Department of Mathematics and Computer Science, University of Southern Denmark, Denmark}{lugovic@imada.sdu.dk}{https://orcid.org/0000-0001-9684-9567}{}

\author{Fabrizio Montesi}{Department of Mathematics and Computer Science, University of Southern Denmark, Denmark}{fmontesi@imada.sdu.dk}{https://orcid.org/0000-0003-4666-901X}{}

\author{Marco Peressotti}{Department of Mathematics and Computer Science, University of Southern Denmark, Denmark}{peressotti@imada.sdu.dk}{https://orcid.org/0000-0002-0243-0480}{}

\authorrunning{L.~Cruz-Filipe et al.}

\Copyright{Luís Cruz-Filipe, Eva Graversen, Lovro Lugović, Fabrizio Montesi, and Marco Peressotti}



\hideLIPIcs
\nolinenumbers

\EventEditors{}
\EventNoEds{0}
\EventLongTitle{}
\EventShortTitle{}
\EventAcronym{}
\EventYear{}
\EventDate{}
\EventLocation{Online}
\EventLogo{}
\SeriesVolume{}
\ArticleNo{}

\begin{document}

\maketitle

\begin{abstract}
Choreographic programming is an emerging programming paradigm for concurrent and distributed systems, whereby developers write the communications that should be enacted and then a distributed implementation is automatically obtained by means of a compiler.
Theories of choreographic programming typically come with strong theoretical guarantees about the compilation process, most notably: the generated implementations operationally correspond to their source choreographies and are deadlock-free.

Currently, the most advanced incarnation of the paradigm is Choral, an object-oriented choreographic programming language that targets Java.
Choral deviated significantly from known theories of choreographies, and in particular introduced the possibility of expressing higher-order choreographies (choreographies parameterised over choreographies) that are fully distributed.
As a consequence, it is unclear whether the usual guarantees of choreographic programming can still hold in the more general setting of higher-order choreographies.

In this article, we introduce \chorlam, the first functional choreographic programming language: it introduces a new formulation of the standard communication primitive found in choreographies as a function, and it is based upon the $\lambda$-calculus.
\chorlam is the first theory that explains the core ideas of higher-order choreographic programming (as in Choral). Interestingly, bridging the gap between practice and theory requires developing a new evaluation strategy and typing discipline for $\lambda$ terms that accounts for the distributed nature of computation in choreographies.
We illustrate the expressivity of \chorlam with a series of examples, which include also reconstructions of the key examples found in the original presentation of Choral.
Our theory supports all the expected properties of choreographic programming.
By offering the first interpretation of choreographies as terms in $\lambda$-calculus, our work also serves to bridge the gap between the communities of functional and choreographic programming.
\end{abstract}

\section{Introduction}
\label{sec:introduction}

\subparagraph*{Choreographic Programming}
Choreographies are coordination plans for concurrent and distributed systems, which prescribe the communications that system participants should enact in order to interact correctly with each other. They are widely used in industry, especially for documentation \cite{msc,bpmn,wscdl}.
Essentially, choreographies are structured compositions of communications. A communication is expressed in one or another variation of the communication term from security protocol notation,
\[
\rolef{Alice} \mathbin{\texttt{->}} \rolef{Bob}\colon M
,
\]
which reads ``$\rolef{Alice}$ communicates the message $M$ to $\rolef{Bob}$'' \cite{NS78}.

Implementing choreographies is notoriously hard, because of the usual issues of concurrent programming. Notably, it requires predicting how processes will interact at runtime, for which programmers do not receive adequate help from mainstream programming technology~\cite{LLLG16,LPSZ08,O18}.
This challenge spawned a prolific area of theoretical research on how choreographies can be related to abstract models of communicating processes, including formal translations of choreographies into process calculi \cite{QZCY07,BZ07,CHY12,Aetal16,Hetal16}.

\emph{Choreographic programming} is an emerging programming paradigm that applies these ideas directly \cite{M13p,Hetal16,GMPRSW21}. In this paradigm programs are choreographies, where communications can be structured by using standard control-flow constructs, e.g., conditionals.
Given a choreography, a distributed implementation can then be automatically produced by means of a compiler, which projects the choreography to a program for each participant that enacts correct message passing behaviour. That is, programs that send and receive messages such that the communications prescribed in the initial choreography are implemented.
We illustrate the choreographic programming method in \cref{fig:chorintro}.
A theory of compilation for choreographies is typically called \emph{Endpoint Projection} (EPP for short)~\cite{CHY12,CM13}.

Choreographic programming has been investigated and showed promise in numerous contexts, including parallel algorithms~\cite{CM16}, cyber-physical systems~\cite{LNN16,LH17}, self-adaptive systems~\cite{DGGLM17}, system integration~\cite{GLR18}, security protocols~\cite{LN15,BCGMS21}, and full functional correctness of distributed programs~\cite{JV22}.
The key reason for its success is that the theoretical foundations of compilation are well-understood and provide strong guarantees.
Specifically, theories of choreographic programming languages usually come with an operational correspondence result between choreographies and the distributed code compiled from them.
A hallmark consequence of this result is that the produced distributed code is deadlock-free.
This property, which stems from the syntactic design of choreographies, was popularised with the slogan of \emph{deadlock-freedom by design}~\cite{CM13}.

\begin{figure}
\centering
\begin{tikzpicture}
\node [
	rectangle, rounded corners, draw,
	minimum width=2cm,
	minimum height=3cm,
	label=above:Choreography with $n$ participants,
	align=left,
] (chor) {$A \mathbin{\texttt{->}} B:x;$ \\ $A \mathbin{\texttt{->}} C:y;$ \\ $C$ computes $z$; \\ $C \mathbin{\texttt{->}} B:z;$ \\ \dots};

\node [
	rectangle, rounded corners, draw,
	below=0.5cm of chor,
	minimum width=8cm,
	minimum height=1.5cm,
] (proj) {Projection};

\node [
	rectangle, rounded corners, draw,
	below=0.5cm of proj.south west,
	minimum width=3.5cm,
	minimum height=2cm,
	label=below:Code for participant A,
	align=left,
] (A) {send $x$ to $B$;\\ send $y$ to $C$; \\ \dots};

\node [
	below=1.25cm of proj,
] (dot) {\textbf{\dots }};

\node [
	rectangle, rounded corners, draw,
	below=0.5cm of proj.south east,
	minimum width=3.5cm,
	minimum height=2cm,
	label=below:Code for participant $n$,
] (N) {projected behaviour};

\draw[-{Latex[length=1.5mm,width=3mm]}] (chor.south) to (A);
\draw[-{Latex[length=1.5mm,width=3mm]}] (chor.south) to (N);

\end{tikzpicture}
\caption{A choreography projected onto its participants (illustration adapted from \cite{GMPRSW21}).}\label{fig:chorintro}
\end{figure}
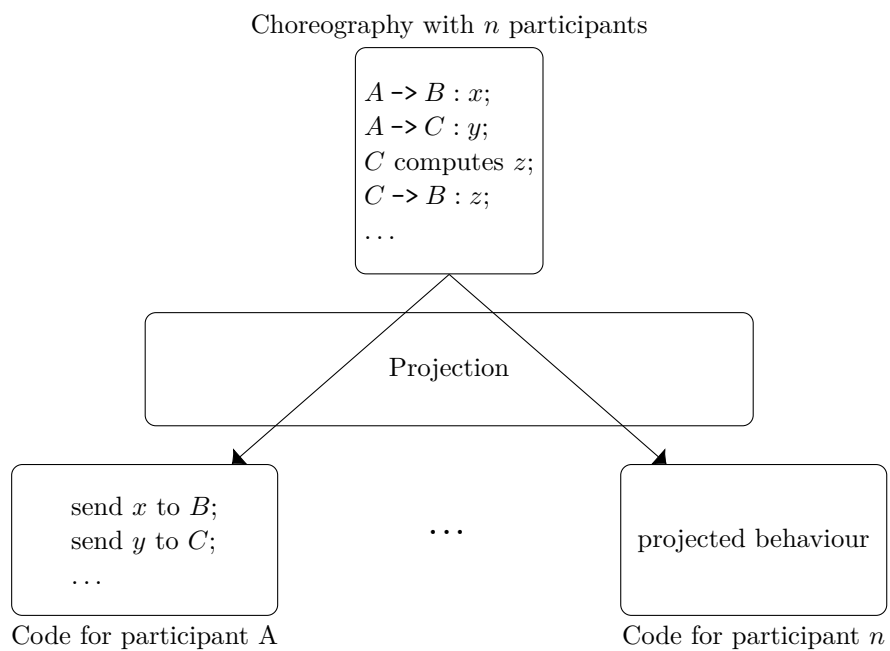

\subparagraph*{Choreographies and Modularity}
In practice, choreographies are large---some even over a hundred pages of text~\cite{openid}.
Thus, it is important to understand the principles of how choreographies can be made modular, enabling the writing (preferably disciplined by types) of large choreographies as compositions of smaller, reusable ones.

The state of the art for modularity in choreographic programming (and choreographic notations in general) is currently represented by Choral, an object-oriented programming language in which choreographies are compiled to Java libraries that applications can use as protocol implementations~\cite{GMP20}.
Choral is the first programming language that is powerful enough to support realistic, mainstream software development with the choreographic method.
In particular, Choral introduced higher-order composition, i.e., the capabilities of defining and invoking choreographies parameterised over other choreographies.
Higher-order composition is essential to many practical scenarios, for instance in the writing of ``extensible'' protocols. An example that we will cover (in \cref{sec:EAP}) is the Extensible Authentication Protocol (EAP), a widely-employed link-layer protocol for the authentication of peers connecting to a network~\cite{eap}. EAP is parametric over a list of authentication protocols, so it is natural to model it as a choreography parametric on a list of other choreographies.

In Choral, data types are equipped with (possibly many) \emph{roles}, which are abstractions of participants.
This allows for writing object methods that involve multiple roles (choreographic methods).
We report an example about this from~\cite{GMP20}.
\begin{example}[Authentication protocol in Choral \cite{GMP20}]\label{ex:authchor}
Consider the case of a distributed authentication protocol, in which a client ($C$) wishes to use its account at an identity provider ($I$) to access a service ($S$). Such a protocol can be implemented in Choral as follows.
\begin{lstlisting}[mathescape=true]
$\funf{class}~\funf{Authenticator}@(S, C, I)\{~\typef{AuthResult}@(C,S)~\funf{authenticate}(\typef{Credentials}@C ~credentials)\{ ... \}~\}$
\end{lstlisting}
In the Choral code above, class \funf{Authenticator} is located at the three roles $S$, $C$, and $I$. Authentication is performed by method $\funf{authenticate}$, which takes the credentials of $C$ (to access its account) and returns the result of the authentication computed at $I$ to $C$ and $S$. Communication in the method body is achieved by invoking special communication methods that can move data from a role to another (the object-oriented interpretation of the communication term of choreographies). The authentication result $\funf{AuthResult}@(C,S)$, is a pair of possibly empty (if the authentication fails) session tokens, with one located at $C$ and the other at $S$.
We will return to this example in \cref{sec:auth}.
\end{example}

\subparagraph*{The Problem}
While the design of Choral is driven by practice, the language also includes features that are not covered by existing theory of choreographic programming.
In particular, we lack an understanding of the key theoretical aspects regarding typing, semantics, and projection for higher-order choreographies. Whether the strong properties expected by choreographic programming still hold in this setting remains unclear.
The aim of this paper is to close this gap.

\subparagraph*{This Article}\label{sec:This}
We present the choreographic $\lambda$-calculus, \chorlam for short, a new theory of choreographic programming that supports higher-order, modular composition.

\chorlam is the first choreographic programming model based on $\lambda$-calculus. This brings two advantages. First, we can tap on a well-known foundation for higher-order programming. Second, it reveals that the key design features of Choral work in the context of functional programming as well. In this way, \chorlam is also the first instance of \emph{functional} choreographic programming.

\chorlam is expressive enough to serve as a model of the core features of Choral, which we illustrate with a series of examples of increasing complexity. In particular, in \cref{sec:chorlam,sec:example}, we recreate in our functional setting some of the key examples given as motivation in the original presentation of Choral (including remote computation, secure key exchange, and single sign-on)~\cite{GMP20}. We also model a more sophisticated scenario based on the Extensible Authentication Protocol (EAP).
Our examples demonstrate that \chorlam allows for parameterising choreographies over different communication semantics, enabling protocol layering, a first for theory of choreographic programming.

In order to capture the essence of higher-order choreographies in the $\lambda$-calculus, we extend its syntax with features from choreographies and ambient calculi (\cref{sec:syntax})~\cite{M22,CG00}. Namely, in \chorlam, data has explicit locations and can be moved between roles by using dedicated communication primitives. For the first time in theory of choreographies, the term for performing a communication is a function and can therefore be composed with other terms as usual in functional programming.

We develop a typing discipline for \chorlam, where types are located at roles (\cref{sec:typing}).
The key novelty of our type system is that it tracks which roles are involved in terms, which requires extending the standard connective for typing functions and a dedicated environment in typing judgements. This is crucial to formulate the expected theory of projection of choreographic programming, which requires to know statically which roles are involved in the choreography being projected.

\chorlam is equipped with an operational semantics for the execution of choreographies (\cref{sec:semantics}). Formulating an appropriate semantics has been particularly challenging, because there is no prior evaluation strategy for the $\lambda$-calculus that is suitable for functional choreographies.
Since choreographies express distributed computation, theories of choreographic languages typically support out-of-order execution for subterms that can be evaluated at independent locations~\cite{CM13}.
How to formulate the necessary inference rules is well-known in the imperative setting, but it has never been studied in others.
We exemplify why an evaluation strategy akin to full $\beta$-reduction is needed (\cref{ex:fullbeta}), but also that standard full $\beta$-reduction cannot be adopted as is (\cref{ex:comorder}).
Our solution is a role-aware evaluation strategy that allows for independent computations to be evaluated in any order, but at the same time enforces an ordering of terms that might interfere with each other (communications).
Well-typed choreographies never get stuck (\cref{thm:TypeReduc}).

We reap the benefits of our design in \cref{sec:epp}, where we define our notion of Endpoint Projection (EPP): a translation from \chorlam to implementations of roles given in terms of a distributed $\lambda$-calculus.
We prove that our EPP is adequate, in the sense that it generates actually distributed implementations: the code generated for a role remains unaffected if the choreography is extended with terms that involve other roles (\cref{thm:RedToUnit}).
Then, we demonstrate that EPP is correct: it generates compliant implementations, in the sense of an operational correspondence (\cref{thm:ChorToNet,thm:NetToChor}).
Thanks to this result and the properties of our type system for choreographies, we obtain that projections of well-typed choreographies are deadlock-free (\cref{thm:NetReduc}).

Our results show that higher-order choreographic programming, as in Choral, has solid theoretical foundations and can as such be adopted by other choreographic languages.

\subparagraph*{Structure of the paper.}
\chorlam, along with its typing and semantics, is presented in \cref{sec:chorlam}. Examples of choreographies inspired by practice are given in \Cref{sec:example}. EPP is discussed in \cref{sec:epp}. Related work is given in \cref{sec:related}. Conclusions are presented in \cref{sec:conc}.

\section{The Choreographic $\lambda$-calculus}
\label{sec:chorlam}
In this section we introduce the Choreographic $\lambda$-calculus, \chorlam. This calculus extends the simply typed $\lambda$-calculus~\cite{Church32} with recursion, choreographic terms for communication, and roles.

The key new features of \chorlam not present in either the $\lambda$-calculus or previous choreographic formalisations are: 1) function types which include every role involved in the function not just the types of the input and output, 2) rewriting rules which move more of the computation into an application for out-of-order execution, and 3) communicating roles added to the labels of transitions to ensure communications happen in the intended order. We will show why these features are crucial for resolving the challenges of higher-order communication.

Roles represent independent participants in a system based on message passing.
Terms in \chorlam are located at roles, to reflect distribution.
For example, the value $\litAt{5}{\rolef{Alice}}$ reads ``the integer $5$ at $\rolef{Alice}$''.
Terms are typed with novel data types that are annotated with roles.
In this case, $\litAt{5}{\rolef{Alice}}$ has the type $\IntAt{\rolef{Alice}}$, read ``an integer at $\rolef{Alice}$''. We write type assignments in the usual way: $\litAt{5}{\rolef{Alice}} : \IntAt{\rolef{Alice}}$.

Values can be moved from a role to another by using a communication primitive, for example, the term $\com{\rolef{Alice}}{\rolef{Bob}} ~\litAt{5}{\rolef{Alice}}$ represents the communication of the value $5$ from $\rolef{Alice}$ to $\rolef{Bob}$.
Therefore, the term evaluates to $\litAt{5}{\rolef{Bob}}$ and has type $\IntAt{\rolef{Bob}}$.

\subsection{Syntax}
\label{sec:syntax}
\begin{definition}
The syntax of \chorlam is given by the following grammar
\begin{align*}
M &\Coloneqq V \mid f(\vec{\role R}) \mid M~M\mid \case M x M x M \mid \select{\role R}{\role R}~l~M\\
V &\Coloneqq x \mid \lambda{x:T}.M \mid \Inl~V \mid \Inr~V \mid \fst \mid \snd \mid \Pair~V~V \mid \litAt{\unit}{\role R} \mid \com{\role R}{\role R} \\
T &\Coloneqq T\rightarrow_{\rho} T \mid T+T \mid T\times T \mid \UnitAt{\role R} \mid \typeAt{t}{\vec{\role R}} %
\end{align*}
where $M$ is a choreography, $V$ is a value, $T$ is a type, $x$ is a variable, $l$ is a label, $f$ is a choreography name (or function name), $\role R$ is a role, $\rho$ is a set of roles, and $t$ is a type variable.
\end{definition}

Abstraction $\lambda x:T.M$, variable $x$ and application $M M$ are as in the standard (simply typed) $\lambda$-calculus.
Likewise for pairs and sums. 
For simplicity, constructors for sums ($\Inl$ and $\Inr$) and products ($\Pair$) are only allowed to take values as inputs, but this is only an apparent restriction: we can define, e.g., a function $\inl$ as $\lambda{x:T}.\Inl~x$ and then apply it to any choreography. Similarly, we define the functions $\inr$ and $\pair$ (the latter is for constructing pairs).
We use these utility functions in our examples.
Sums and products are deconstructed in the usual way, respectively by the $\ccase$ construct and by the $\fst$ and $\snd$ primitives.

The primitives $\com{S}{R}$ and $\select{S}{R}~l~M$ (where $S$ and $R$ are roles) come from choreographies and are the only primitives of \chorlam that introduce interaction between different roles.
The term $\com{S}{R}$ is a \emph{communication}: it acts as a $\lambda$-expression that takes a value at role $\role{S}$ and returns the same value at role $\role{R}$. The standard choreographic primitive for synchronous communication $\rolef{Alice} \mathbin{\texttt{->}} \rolef{Bob}\colon M$ is recovered as the function application $\com{\rolef{Alice}}{\rolef{Bob}}~M$. 
The term $\select{S}{R}~l~M$, instead, is a \emph{selection}, where $\role{S}$ informs $\role{R}$ that it has selected the label~$l$ before continuing as $M$. Selections choreographically represent the communication of an internal choice made by $\role{S}$ to $\role{R}$, the effects of which are not visible on a choreographic level, as the choreography knows the result of all internal choices. (In the implementation of choreographies, a selection corresponds to an internal choice at the sender and an external choice at the receiver.)

Finally, $f(\vec{R})$ stands for a (choreographic) function $f$ instantiated with the roles $\vec{R}$,
which evaluates to the body of the function as given by an environment of definitions (a mapping from function names to choreographies).
Function names are used to model recursion.
In the typing and semantics of \chorlam, we will use $D$ to range over mappings of function names to choreographies.
Within a choreography, there is no need to distinguish between roles that are statically fixed and role parameters: inside of a function definition, all roles are parameters of the function; otherwise, all roles are statically determined. All roles are treated in the same way by our theory. By convention, we write function role parameters of functions with uppercase letters as $R$, $R'$, $S$, etc., and fixed roles as $\rolef{Alice}$, $\rolef{Bob}$, etc.

To illustrate base values, we also have a term $\unit@R$ which denotes a unit value at the role $R$---other base values, like $5@R$ used in the examples above, can be easily included following the same approach.
Values are not limited to one role in general; for example, $\Pair~\unit@S~\unit@R$ denotes a distributed pair where the first element resides at $S$ and the second at $R$.
We say a choreography (or value or type) is local to $R$ if $R$ is the only role mentioned in any subterm of the choreography, e.g., $\lambda x:\unitt@R.(\Pair~x~\unit@R)$ is a local function located at $R$.
In examples we will sometimes use the shorthand $M@R$, $V@R$, and $T@R$ to denote a choreography, value, or type which is located entirely at the role $R$, though this is not part of the syntax of \chorlam.

In \chorlam types record the distribution of values across roles: if role $\role{R}$ occurs in the type given to $V$ then part of $V$ will be located at $\role{R}$.
Because a function may involve more roles besides those listed in the types of their input and output, the type of abstractions $T\rightarrow_\rho T'$ is annotated with a set of roles $\rho$ denoting the roles that may participate in the computation of a function with that type besides those occurring in the input $T$ or the output $T'$---in examples we will often omit this annotation if the set of additional roles is empty thus writing $T\rightarrow T'$ instead of $T\rightarrow_\emptyset T'$.
For example, if $\rolef{Alice}$ wants to communicate an integer to $\rolef{Bob}$ directly (without intermediaries), she can use a choreography of type $\IntAt{\rolef{Alice}} \to \IntAt{\rolef{Bob}}$; however, if the communication might go through a proxy, then she can use a choreography of type $\IntAt{\rolef{Alice}} \to_{\{\rolef{Proxy}\}} \IntAt{\rolef{Bob}}$.
As we will see in \cref{sec:epp}, this information is essential for compiling functions that involve multiple roles to distributed implementations that do not require additional coordination wrt what is written in the choreography.

Aside from the annotations on arrows our types resemble those of simply typed $\lambda$-calculus and serve the same primary purpose of keeping track of input and output of functions in order to prevent nonsensical choreographies.
 Consider the function $h$ defined as $\fun{x}{\IntAt{\rolef{Alice}}}\com{\rolef{Proxy}}{\rolef{Bob}}~(\com{\rolef{Alice}}{\rolef{Proxy}}~x)$, which communicates an integer from $\rolef{Alice}$ to $\rolef{Bob}$ by passing through an intermediary $\rolef{Proxy}$ and has the type $\IntAt{\rolef{Alice}} \to_{\{\rolef{Proxy}\}} \IntAt{\rolef{Bob}}$.
 For any term $M$, the composition $h~M$ makes sense if the evaluation of $M$ returns something of the type expected by $h$, that is $\IntAt{\rolef{Alice}}$. The composition $h~\litAt{5}{\rolef{Alice}}$ makes sense, but $h~\litAt{5}{\rolef{Bob}}$ does not, because the argument is not at the role expected by $h$.

The types for sums and products are the usual ones (forming a sum or product of $T$ and $T'$ does not introduce new roles besides those already listed in $T$ and $T'$).
The type of units is annotated with the role where each unit is located; $\unitt@R$ is the type of the unit value available (only) at role $\role{R}$.
Recursive type variable $t@\vec{\role R}$ are annotated with the roles $\vec{R}$, instantiating the roles occurring in their definition (we will discuss type definitions later in this section).
The set of roles in a type is formally defined as follows.

\begin{definition}[Roles of a type]
The roles of a type $T$, $\roles(T)$, are defined as follows.
\begin{align*}
	\roles(t@\vec{\role R}) &= \vec{\role R}
  & \roles(T \rightarrow_\rho T') &=\roles(T)\cup \roles(T')\cup \rho \\
  \roles(\UnitAt{R}) &=\{\role R\}
  & \roles(T + T') = \roles(T\times T') &=\roles(T)\cup \roles(T')
\end{align*}
\end{definition}

\begin{example}[Remote Function]
  \label{ex:remfun}
  We can use \chorlam to define a small choreography, $\funf{remoteFunction}(C, S)$ for a distributed computation in which a client, $C$ sends an integer $\var{val}$ to a server $S$ where a function $\var{fun}$ located at $S$ is applied to $\var{val}$ before the result gets returned to $C$. For simplicity, we assume a primitive type for integers ($\Int$).
\begin{lstlisting}[mathescape=true]
$\fundef{remoteFunction}(C, S) = \fun{fun}{\IntAt{S} \to
\IntAt{S}}\;\fun{val}{\IntAt{C}}\;\com{S}{C}\;(\var{fun}\; (\com{C}{S}\;\var{val}))$
\end{lstlisting}
This choreography is parametrised on the roles $S$ and $C$ as well as the local function $\var{fun}$ and value $\var{val}$.
\eoe
\end{example}

Crucially, a choreographic term $M$ may involve more roles besides those listed in its type. For instance, the choreographies $\litAt{\unit}{R}$, $\com{S}{R}~\litAt{\unit}{S}$, and $\com{P}{R}~\com{S}{P}~\litAt{\unit}{S}$ all have type $\UnitAt{R}$, but they implement different communication behaviours. This makes choreographies compositional, and will be important in establishing typing preservation later.

A key concern of choreographic languages is knowledge of choice: the property that when a choreography chooses between alternative branches (as with our $\ccase$ primitive), all roles that need to behave differently in the branches are properly informed via appropriate selections~\cite{CDP11}. We give an example of how selections should be used, and postpone a formal discussion of how knowledge of choice is checked for to our presentation of Endpoint Projection in \cref{sec:epp}.

\begin{example}[Remote Map]
  \label{ex:choice}
  We now build on the remote function from \cref{ex:remfun} by using it to create a choreography $\funf{remoteMap}(C, S)$, where the server $S$ applies a local function to not just one value received from the client $C$, but instead to each element of a list sent individually from $C$ to $S$ and then returned after the computation at $S$ is complete.
\begin{lstlisting}[mathescape=true]
$\fundef{remoteMap}(C, S) = \fun{fun}{\IntAt{S} \to \IntAt{S}}\;\fun{list}{\ListAt{\Int}{C}}$
  $\ccase\;list\;\cof$
    $\cleft{x} \select{C}{S}\;\lab{stop} \; \litAt{\unit}{C} \csep$
    $\cright{x} \select{C}{S}\;\lab{again}\;$
      $\funf{cons}(C)\; (\funf{remoteFunction}(C, S)\; fun
\; (\fst\;x)) \;(\funf{remoteMap}(C, S)\; fun \; (\snd\;x))$
\end{lstlisting}
Here, $\ListAt{\Int}{C}$ is the recursive type satisfying
$\ListAt{\Int}{C}=\UnitAt{C}+(\IntAt{C}\times\ListAt{\Int}{C})$, representing a list of integers. In general, we write $\ListAt{t}{(R_1,\dots,R_n)}$ to mean the recursive type satisfying 
\[\ListAt{t}{(R_1,\dots,R_n)} = (\UnitAt{R_1} \times \dots \times \UnitAt{R_n}) + (\typeAt{t}{(R_1,\dots,R_n)} \times \ListAt{t}{(R_1,\dots,R_n))}\text.\]
When we introduce typing judgements, we will show how to work with this kind of type equations.

The choreography uses selections so that $S$ is informed about how it should behave (terminate or recur). This is essential if the choreography is to be implemented in a fully distributed way, since the information is initially available only at $C$.
As we will see in \cref{sec:epp}, we can automatically generate the following local behaviours for $C$ and $S$, respectively.
\begin{lstlisting}[mathescape=true]
$\fun{fun}{\botmt \to \botmt}\;\fun{list}{\List{\Int}}$
  $\ccase\;list\;\cof$
    $\cleft{x} \oplus_{S}\;\lab{stop} \; \unit \csep$
    $\cright{x} \oplus_{S}\;\lab{again}\;$
      $\funf{cons}\; (\funf{remoteFunction}_1(S)\; \botm
\; (\fst\;x)) \;(\funf{remoteMap}_1(S)\; \botm \; (\snd\;x))$
\end{lstlisting}
\begin{lstlisting}[mathescape=true]
$\fun{fun}{\typef{Int} \to \typef{Int}}\;\fun{list}{\botmt}$
  $\&_C\{\lab{stop}: \; \botm; $
    $\;\lab{again}: \; (\funf{remoteFunction}_2(C)\; fun
\; \botm) \;(\funf{remoteMap}_2(C)\; fun \; \botm)\}$
\end{lstlisting}

In the local behaviours, term $\botm$ denotes a part of the computation that takes place elsewhere.
Notice how the $\ccase$ is evaluated on data at role $\role{C}$, so that role is the only one initially knowing which branch has been chosen. Each branch, however, starts with role $\role{C}$ sending a label to role $\role{S}$, using term $\oplus_{S}$.
On the other hand, $S$ does not have a $\ccase$ but must wait to receive a label from $C$, using term $\&_C$, to figure out whether it should terminate (label $\lab{stop}$) or continue (label $\lab{again}$): from its point of view, $S$ is reactively handling a stream.

The check for knowledge of choice that we will formalise in \cref{sec:epp} (the merge operator) supports two principles: roles do not need to receive a selection until their behaviour is affected by a choice made by another role; and knowledge of choice can be propagated, in the sense that any role that has been informed of a choice through a selection can inform other roles as well (selections do not necessarily come from the role that has made the choice).
\eoe
\end{example}

We can encode conditional statements in the standard way: we define a type $\BoolAt{R}$ as $\UnitAt{R}+\UnitAt{R}$, and $\cif \; M \; \cthen \; M' \; \celse \; M''$ as an abbreviation for $\case Mx{M'}x{M''}$.

Free and bound variables are defined as expected, noting that $x$ and $y$ are bound in $\case Mx{M'}y{M''}$. 
We write $\fv(M)$ for the set of free variables in term $M$. 
The formal definition is given in \cref{sec:appendix}. We call a choreography closed if it has no free variables and restrict our results to closed choreographies. 

\subsection{Typing}
\label{sec:typing}

We now show how to type choreographies following the intuitions already given earlier.
Typing judgements have the form $\Theta;\Sigma;\Gamma\vdash M:T$, where: 
$\Theta$ is the set of roles that can be used for typing $M$; 
$\Sigma$ is collection of type definitions parameterised on roles, i.e., expressions of the form $t@\vec{R}=T$ where the elements of $\vec{R}$ are distinct; and
$\Gamma$ is a typing environment assigning variable names to their type ($x:T$) and choreography names to the their set of bound roles and type (written $f(\vec{R}):T$).
We further require that a type variable $t$ is defined at most once in $\Sigma$,
that definitions are contractive, and that $\roles(T)=\vec{R}$ for any $t@\vec{R}=T\in\Sigma$. We can use $\Sigma$ to define common types such as $\BoolAt{ R }=(\UnitAt{R}+\UnitAt{R})$. It also allows us to define recursive types such as $\ListAt{\Bool}{R}=\UnitAt{R}+(\BoolAt{R} \times \ListAt{\Bool}{R})$ as described in \cref{ex:choice}.
We call $\Theta;\Sigma;\Gamma$ jointly a \emph{typing context}.
We also require that variable and choreography names occur in $\Gamma$ at most once. %
Many of the rules resemble those for simply typed $\lambda$-calculus, but with roles added, and the additional requirements that only the roles in the type are used in the term being typed.
We include some representative ones in \cref{fig:typing} (the complete set of typing rules is given in \cref{sec:appendix}). Our typing rules use the predicate $\dist(\vec{R})$ to indicate that the elements of $\vec{R}$ are distinct and $||\vec{R}||$ to denote the number of elements of $\vec{R}$.

The most interesting part of our type system is the annotation $\rho$ on the function type $T\rightarrow_\rho T'$, which is essential for projection. Without it, we would not know which roles might be involved in a function, because the function's type would show only the roles in its input and output types. This knowledge is essential for projecting an application exactly to the roles that need to be involved in the computation. 
Consider the choreography $(\fun{f}{\IntAt{S} \to_{\{S'\}} \IntAt{R}}{f~\litAt{5}{S}})~(\fun{x}{\IntAt{S}}{(\com{S}{S’}~(\com{S’}{R}~x))})$. 
Without the annotation $\{S'\}$ on the type of $f$, the projection has no way to know that the left part of this application will involve, and therefore need to be projected at, $S’$. 
This is essential for modularity of projection and reusability of the generated code: we can project functions that can be used with arguments supplied by any (well-typed) context.
\Cref{rule:t-abs} uses $\Theta$ to ensure that $\rho$ contains every additional role used in the function by requiring every role to be in $\Theta$ and restricting $\Theta$ to the roles of $T$, $\rho$, and $T'$.%

\begin{figure}
\begin{spreadlines}{\rulesvsep}
\begin{eqnarray*}&
  \infer[\rtag{\rlabel{\rname{TVar}}{rule:t-var}}]
		{\Theta;\Sigma;\Gamma\vdash x:T}
		{x:T\in \Gamma & \roles(T)\subseteq \Theta}
	 \qquad \infer[\rtag{\rlabel{\rname{TApp}}{rule:t-app}}]
	{\Theta;\Sigma;\Gamma\vdash N~M:T'}
	{\Theta;\Sigma;\Gamma\vdash N:T\rightarrow_\rho T' 
	&\Theta;\Sigma;\Gamma\vdash M:T}
	 &\\&
	\infer[\rtag{\rlabel{\rname{TDef}}{rule:t-fun}}]
		{\Theta;\Sigma;\Gamma\vdash f(\vec{\role R}):T[\vec{R'}:=\vec{R}]}
		{f(\vec{\role{R'}}):T\in \Gamma & \roles(T)\subseteq \vec{\role R}\subseteq\Theta & ||\vec{\role R}||=||\vec{\role R'}|| & \dist(\vec{R})} 
	&\\&
  \infer[\rtag{\rlabel{\rname{TAbs}}{rule:t-abs}}]
  	{\Theta;\Sigma;\Gamma\vdash\lambda x:T.M:T \rightarrow_\rho T'}
  	{\Theta';\Sigma;\Gamma,x:T\vdash M:T' & \rho\cup \roles(T)\cup \roles(T') = \Theta' \subseteq \Theta}
	&\\&
  \infer[\rtag{\rlabel{\rname{TCom}}{rule:t-com}}]
		{\Theta;\Sigma;\Gamma\vdash \com{S}{R}:T \rightarrow_\emptyset T[\role{S}:=\role{R}]}
		{\roles(T)=\{\role{S}\} & \{S,R\}\subseteq \Theta}
	\qquad
	\infer[\rtag{\rlabel{\rname{TSel}}{rule:t-sel}}]
		{\Theta;\Sigma;\Gamma\vdash \select{S}{R}~l~M:T}
		{\Theta;\Sigma;\Gamma\vdash M:T & \{S,R\}\subseteq \Theta}
	&\\&
	\infer[\rtag{\rlabel{\rname{TEq}}{rule:t-eq}}]
			{\Theta;\Sigma;\Gamma\vdash M:T[{\vec{R}}:={\vec{R'}}]}
			{\Theta;\Sigma;\Gamma\vdash M:t@\vec{R'} 
			&t@\vec{R}=_{\Sigma} T & \vec{R'} \subseteq \Theta &
      ||\vec{\role R}||=||\vec{\role R'}|| &
      \dist(\vec{R'})}
&\end{eqnarray*}
\end{spreadlines}
\caption{Typing rules for \chorlam (representative selection).}
\label{fig:typing}
\end{figure}

Otherwise, \cref{rule:t-var,,rule:t-fun,rule:t-abs} exemplify how role checks are added to the standard typing rules for simply typed $\lambda$-calculus.
\Cref{rule:t-com} types communication actions, moving subterms that were placed at role $\role{S}$ to role~$\role{R}$ (here $T[S := R]$ is the type expression obtained by replacing $S$ with $R$). Note that the type of the value being communicated must be located entirely at $S$.
\Cref{rule:t-sel} types selections as no-ops, only checking that the sender and receiver of the selection are legal roles.
\Cref{rule:t-eq} allows rewriting a type according to $\Sigma$ in order to mimic recursive types (see \cref{ex:choice}).

We also write $\Theta;\Sigma;\Gamma\vdash D$ to denote that a set of definitions $D$, mapping names to choreographies, is well-typed.
Sets of definitions play a key role in the semantics of choreographies, and can be typed by the rule below.
\[
	\infer[\rtag{\rlabel{\rname{TDefs}}{rule:t-defs}}]
		{\Theta;\Sigma;\Gamma\vdash D}
		{\forall f(\vec{R})\in \mathsf{domain}(D): 
		& f(\vec{R}):T\in \Gamma 
		& \vec{R};\Sigma;\Gamma\vdash D(f(\vec{R})):T & \dist(\vec{R}) & \vec{R}\subseteq\Theta}
\]

\begin{example}
  The set of definitions in \cref{ex:remfun,ex:choice} can be typed in the typing context:
  \begin{align*}
   \Theta &= \{C,S\}
   \\
   \Sigma &=\{\ListAt{\Int}{R}=\UnitAt{R}+(\IntAt{R}\times\ListAt{\Int}{R})\}
   \\
   \Gamma &=
   \left\{\begin{aligned}
   \funf{remoteFunction}& : (\IntAt{S}\to\IntAt{S})\to\IntAt{C}\to\IntAt{C},\\
   \funf{remoteMap}& : (\IntAt{S}\to\IntAt{S})\to\ListAt{\Int}{C}\to\ListAt{\Int}{C}
   \end{aligned}\right\}
  \end{align*}
  \eoe
\end{example}

\subsection{Semantics}
\label{sec:semantics}
\chorlam comes with a reduction semantics that captures the essential ingredients of the calculi that inspired it: $\beta$- and $\iota$-reduction, from $\lambda$-calculus, and the usual reduction rules for communications and selections.
Some representative rules are given in \cref{fig:semantics}.

\begin{figure}
\begin{spreadlines}{\rulesvsep}
  \begin{eqnarray*}&
    \infer[{\rtag{\rlabel{\rname{AppAbs}}{rule:c-appabs}}}]
	    {\lambda x:T.M~V\xrightarrow{\tau,\emptyset}_D M[x:=V]}
      {}
    &\\& 
     \infer[\rtag{\rlabel{\rname{InAbs}}{rule:c-inabs}}]
     	{\lambda x:T.M\xrightarrow{\lambda,\mathbf{R}}_D \lambda x:T.M'}
     	{M\xrightarrow{\ell,\mathbf{R}}_D M' }
    &\\&
     \infer[\rtag{\rlabel{\rname{App1}}{rule:c-app1}}]
     	{M~N\xrightarrow{\tau,\mathbf{R}}_D M'~N}
     	{M\xrightarrow{\ell,\mathbf{R}}_D M' & \ell=\lambda\Rightarrow\mathbf{R}\cap \roles(N)=\emptyset}
    &\\&
     \infer[\rtag{\rlabel{\rname{App2}}{rule:c-app2}}]
     {V~N\xrightarrow{\tau,\mathbf{R}}_D V~N'}
     {N\xrightarrow{\tau,\mathbf{R}}_D N'} 
     \qquad  
     \infer[\rtag{\rlabel{\rname{App3}}{rule:c-app3}}]
     {M~N\xrightarrow{\tau,\mathbf{R}}_D M~N'}
     {N\xrightarrow{\tau,\mathbf{R}}_D N' & \mathbf{R}\cap\roles(M)=\emptyset} 
    &\\&
     \infer[\rtag{\rlabel{\rname{Case}}{rule:c-case}}]
     	{\case{N}{x}{M}{x'}{M'}\xrightarrow{\tau,\mathbf{R}}_D \case{N'}{x}{M}{x'}{M'}}
     	{N\xrightarrow{\tau,\mathbf{R}}_D N'} 
    &\\&
     \infer[\rtag{\rlabel{\rname{InCase}}{rule:c-incase}}]
     	{\begin{array}{l}\case{N}{x}{M_1}{x'}{M_2}\xrightarrow{\ell, \mathbf{R}}_D \\ \hphantom{\case{N}{x}{M_1}{x'}{M_2}} \case{N}{x}{M_1'}{x'}{M_2'}\end{array}}
     	{M_1\xrightarrow{\ell,\mathbf{R}}_D M_1' & M_2\xrightarrow{\ell,\mathbf{R}}_D M_2' & \mathbf{R}\cap \roles(N)=\emptyset} 
    &\\&
      \infer[\rtag{\rlabel{\rname{CaseL}}{rule:c-casel}}]
       {\case{\Inl~V}{x}{M}{x'}{M'}\xrightarrow{\tau,\emptyset}_D M[x:= V]}
       {}
    &\\&
      \infer[\rtag{\rlabel{\rname{Proj1}}{rule:c-proj1}}]
     {\fst~\Pair~V~V'\xrightarrow{\tau,\emptyset}_D V}
     {}
   \qquad
	    \infer[\rtag{\rlabel{\rname{Def}}{rule:c-fun}}]{f(\vec{R}) \xrightarrow{\tau,\emptyset}_D M [{\vec{R'}}:={\vec{R}}]}{D(f(\vec{R'}))=M}
    &\\&
     \infer[{\rtag{\rlabel{\rname{Com}}{rule:c-com}}}]{\com{S}{R}~V \xrightarrow{\tau,\{S,R\}}_D V[S:=R]}{\fv(V)=\emptyset}
    \qquad
      \infer[\rtag{\rlabel{\rname{Sel}}{rule:c-sel}} ]
     {\select{S}{R}~l~M \xrightarrow{\tau,\{S,R\}}_D M}
     {}
     &\\&
     \infer[\rtag{\rlabel{\rname{InSel}}{rule:c-insel}}]
     {\select{S}{R}~\ell~M \xrightarrow{\ell,\mathbf{R}}_D \select{S}{R}~\ell~M'}
     {M\xrightarrow{\ell,\mathbf{R}}_D M' & \mathbf{R}\cap \{S,R\}=\emptyset} 
     \qquad
     \infer[\rtag{\rlabel{\rname{Str}}{rule:c-str}}]{M\xrightarrow{\tau,\mathbf{R}}_D M'}{M\rightsquigarrow^* N & N\xrightarrow{\tau,\mathbf{R}} M'}
  \end{eqnarray*}
  \end{spreadlines}
  \caption{Semantics of \chorlam.}
  \label{fig:semantics}
\end{figure}
\begin{figure}
\begin{spreadlines}{\rulesvsep}
  \begin{eqnarray*}&
\infer[{\rtag{\rlabel{\rname{R-AbsR}}{rule:r-absr}}}]{((\lambda x:T.M)~N)~M'\rightsquigarrow (\lambda x:T.(M~M'))~N)}{x\notin \fv(M')} &\\& 
\infer[{\rtag{\rlabel{\rname{R-AbsL}}{rule:r-absl}}}]{M'~((\lambda x:T.M)~N)\rightsquigarrow (\lambda x:T.(M'~M))~N)}{x\notin \fv(M') & \sroles(M')\cap\roles(N)=\emptyset}   
&\\&
\infer[{\rtag{\rlabel{\rname{R-CaseR}}{rule:r-caser}}}]{\begin{array}{l}(\case{N}{x}{M_1}{x'}{M_2})~M\rightsquigarrow \\
\qquad\qquad\qquad\qquad \case{N}{x}{(M_1~M)}{x'}{(M_2~M)}\end{array}}{x,x'\notin \fv(M)}
&\\&
\infer[{\rtag{\rlabel{\rname{R-CaseL}}{rule:r-casel}}}]{\begin{array}{l}M~(\case{N}{x}{M_1}{x'}{M_2})\rightsquigarrow\\ \qquad\qquad\qquad\qquad \case{N}{x}{(M~M_1)}{x'}{(M~M_2)}\end{array}}{x,x'\notin \fv(M) & \sroles(M)\cap \roles(N)=\emptyset}
&\\&
\infer[{\rtag{\rlabel{\rname{R-SelR}}{rule:r-selr}}}]
{(\select{S}{R}~l~N)~M\rightsquigarrow \select{S}{R}~l~(N~M)}
{}
 &\\&
\infer[{\rtag{\rlabel{\rname{R-SelL}}{rule:r-sell}}}]{M~(\select{S}{R}~l~N)\rightsquigarrow \select{S}{R}~l~(M~N)}{\sroles(M)\cap\roles(N)=\emptyset}
\end{eqnarray*}
  \end{spreadlines}
  \caption{Rewriting of \chorlam.}
  \label{fig:rewrite}
\end{figure}

The key idea of our semantics is that terms at different roles can be evaluated independently, unless interaction is specified within the choreography. This kind of out-of-order execution is typical for choreographic calculi, and we port it to $\lambda$-calculus here for the first time~\cite{CM13}. In addition to functional choreographies having a different structure to imperative, out-of-order execution in higher-order choreographies is complicated by having actions where multiple roles are involved but no synchronisation happens, namely applications of values located at multiple roles such as choreographies and pairs with elements located at different roles.

The semantics are annotated with a label, $\ell$, and a set of synchronising roles, $\mathbf{R}$. The label is either $\lambda$, when an action is propagated out through a $\lambda$-term as in \cref{rule:c-inabs}, or $\tau$ otherwise. The set of synchronising roles is empty if no synchronisations are taking place. The purpose of the label and synchronising roles is to ensure that synchronisations between the same roles occur in the expected order, the importance of which becomes clear later.

\Cref{rule:c-appabs,,rule:c-app1,,rule:c-app2} implement a call-by-value $\lambda$-calculus.
\Cref{rule:c-case,,rule:c-casel} and its counterpart~\cref{rule:c-caser} implement $\iota$-reductions for sums, and likewise for \cref{rule:c-proj1,,rule:c-proj2} wrt pairs.
The communication \cref{rule:c-com} changes the associated role of a value, moving it from $\role{S}$ to $\role{R}$, while the selection \cref{rule:c-sel} implements selection as a no-op.
\Cref{rule:c-fun} allows reductions to use choreographies defined in $D$.

In addition to the fairly standard $\lambda$-calculus semantics, we have some rules for out-of-order execution. These include rewriting terms as described in \cref{fig:rewrite} and being able to propagate some transitions past an abstraction, case, and selection as in \cref{rule:c-incase,rule:c-inabs,rule:c-insel}. We also have a ``role-aware'' variation of full $\beta$-reduction by using \cref{rule:c-app3}, the need for which is illustrated by \cref{ex:fullbeta}.
These rules serve the purpose of making our semantics decentralised, in the sense that actions at distinct roles can proceed independently.

\begin{example}\label{ex:fullbeta}
Consider the choreography $M=f(S)~((\lambda x:T@R.V@S)~V'@R)$.
The choreography includes two independent roles, $R$ and $S$, but the two never actually interact: the inner application involves an abstraction and an argument located only at $R$, so it should be evaluated independently of $R'$. Likewise, $f(S)$ is located entirely at role $S$, so it should be evaluated independently of $R$.

Without \cref{rule:c-app3}, $M$ would be unable to evaluate the inner application before $f(S)$ finished running, which may be never if $f$ diverges, breaking the assumption that roles execute in a decentralised way.
\eoe
\end{example}

The rewriting rules are not standard for the $\lambda$-calculus, but they are not as strange as they first appear. Take, for example, \cref{rule:r-absr}; it simply states that if you have a function with two variables, $\lambda x:T.\lambda y:T'.M$, then $x$ and $y$ can be instantiated in any order as long as they each get the correct value. On the other hand, \cref{rule:r-absl} says that more of the computation can be pushed into an abstraction so long as it does not affect the order of synchronisations. The other rewriting rules work on similar principles, but dealing with conditionals and selections. These rules all work to ensure that while actions can be performed in different orders the result of the computation must remain the same before and after rewriting.
In \cref{ex:rewrite} we see why we need the rewriting rules in order to support the out-of-order executions necessary for the choreography to allow concurrent execution of computations located at different roles.

\begin{figure}
\centering
\begin{tikzpicture}[node distance=1cm,auto]  
\tikzset{ 
}  
\node (start) at (0,0) {$((\lambda x:\unitt@R.\lambda x':T@S.\unit@S)~f(R))~V@S$};  
\node (r1) at (0,-1.5) {$((\lambda x:\unitt@R.(\lambda x:T@S.\unit@S)~V@S)~f(R))$}; 
\node (r2) at (0,-3) {$(\lambda x:\unitt@R.\unit@S)~f(R)$};

\draw[->,squiggly] (start) to node {Rewriting with \cref{rule:r-absr}}  (r1);	
\draw[->] (r1) to node {Application with \cref{rule:c-app3,rule:c-inabs,rule:c-appabs}}  (r2);

\end{tikzpicture} 
\caption{An example of rewriting (\cref{ex:rewrite}).}\label{fig:rewrite-example}
\end{figure}

\begin{example}[Rewriting]\label{ex:rewrite}
Consider the choreography with an abstraction at $S$ inside an abstraction at $R$, $M=((\lambda x:\unitt@R.\lambda x':T@S.\unit@S)~f(R))~V@S$. As in \cref{ex:fullbeta}, $R$ and $S$ each independently execute their part of the choreography. $R$ evaluates $f(R)$ and then applies the result. Independently, $S$, executes the other application $\lambda x':T@S.\unit@S~V@S$. For $M$ to be able to execute the application at $S$ independently of $R$'s actions, we need \cref{rule:r-absr} to get $\lambda x:T@S.\unit@S$ and $V@S$ next to each other by rewriting to $((\lambda x:\unitt@R.(\lambda x:T@S.\unit@S)~V@S)~f(R))$ and \cref{rule:c-inabs} to propagate the application of $(\lambda x:T@S.\unit@S)$ and $V@S$ past $\lambda x:\unitt@R$ as shown in \cref{fig:rewrite-example}.

This problem can also be seen in e.g., $(\lambda y:T@S.M@S)~\case{f(R)}{x}{V@S}{x'}{V@S}$ where, since $f(R)$ is located at $R$, $S$ can execute the application $(\lambda y:T@S.M@S)~V@S$ independently of the result of the result of the conditional. 
\eoe
\end{example}

Some of the out-of-order-execution rules, specifically the ones pushing the left part of an application further in, have restrictions on them because we want to avoid there being more than one communication or synchronisation available at the same time on the same roles. This is because we need to ensure that communications and selections on a specific set of roles must always happen in the same order, as we otherwise get the problems illustrated by \cref{ex:comorder}.
\begin{figure}
\centering
\begin{tikzpicture}[node distance=1cm, auto]  
\node (c0) {$(\lambda x:T@R.(\com{S}{R}~V@S))~(\com{S}{R}~V'@S)$};
\node[below=of c0] (n0) {$R[(\lambda x:T.(\boxed{\recv{S}~\botm}))~(\recv{S}~\botm)]\mid S[(\lambda x:\botmt.(\send{R}~V))~(\boxed{\send{R}~V'})]$};
\node[below=of n0] (n1) {$R[(\lambda x:T.(\boxed{V'}))~(\recv{S}~\botm)]\mid S[(\lambda x:\botmt.(\send{R}~V))~(\boxed{\botm})]$};
\draw[->] (c0) to node {Projection} (n0);
\draw[->] (n0) to node {Synchronisation} (n1);
\end{tikzpicture} 
\caption{An example of ``mismatching'' synchronisation at the network level without a counterpart at the choreographic level (we highlight terms in the behaviours of $S$ and $P$ active in the synchronisation).}
\label{fig:exsynch}
\end{figure}
\begin{example}[Communication order]\label{ex:comorder}
Consider a choreography with two communications between the same roles, $M=(\lambda x:T@R.(\com{S}{R}~V@S))~(\com{S}{R}~V'@S)$. This has a similar structure to $((\lambda x:\unitt@R.(\lambda x':T@S.M)~V@S)~f(R))$ from \cref{ex:rewrite}, with part of the computation hidden behind an abstraction. However, while \cref{ex:rewrite} needed to use \cref{rule:c-inabs} to allow the computation inside of the abstraction to execute before the computation outside, doing so would cause problems in $(\fun{ x}{T@R}(\com{S}{R}~V@S))~(\com{S}{R}~V'@S)$.

Here, the behaviour of $M$ at $R$ is $(\fun{x}{T}(\recv{S}~\botm))~(\recv{S}~\botm)$ and the behaviour at $S$ is $(\fun{x}{\botmt}(\send{R}~V))~(\send{R}~V')$.

Like many previous choreographic languages, our implementation model for choreographies in \cref{sec:epp} assumes that each pair of roles has one channel between them, which they use for all communications. 
Therefore, if the two communications can be performed in any order then $S$ is currently able to send either $V$ or $V'$ and $R$ is correspondingly able to receive either inside or outside the abstraction. Since $S$ and $R$ act independently, we have no guarantee that if $S$ chooses to send $V$ first $R$ will also choose to use its left receive action or vice versa. This can create situations such as the one seen in \cref{fig:exsynch}, which our choreographies cannot and should not model.

We therefore restrict the out-of-order communications by restricting the synchronising names in \cref{rule:c-inabs,rule:c-app1,rule:c-app3,rule:c-incase,rule:c-insel}. Another way to solve the issue would be to define an implementation model where we use channels that differentiate between the two occurrences of $\com{S}{R}$.
\eoe
\end{example}

We therefore have \cref{thm:onesynch} stating that any reductions available at the same time must have different (or no) synchronisation roles.
\begin{proposition}\label{thm:onesynch}
Given a choreography $M$, if $M\xrightarrow{\ell,\mathbf{R}} M'$ and $M \xrightarrow{\ell',\mathbf{R}'} M''$ and there does not exist $N$ such that $M'\rightsquigarrow^* N$, and $M''\rightsquigarrow^* N$, then $\mathbf{R}\cap\mathbf{R'}=\emptyset$. 
\end{proposition}
\begin{proof}
The key here is that unless these transitions are either communications or selections, $\mathbf{R}$ and $\mathbf{R'}$ are empty. Once this is clear, the rest follows by induction on $M$.
\end{proof}

We also use the label $\lambda$ in \cref{rule:c-inabs} to restrict these out-of-order communications, since we do not know which roles we need to restrict communication on in \cref{ex:comorder} until we reach the application, at which point the $\lambda$ label becomes a $\tau$ again if it is allowed to propagate. 

Since we restrict out-of-order communication in the other out-of-order execution rules, we need to be the same in the rewriting rules as shown in \cref{ex:rewritecom}. For this purpose we use the concept of synchronisation roles.

\begin{definition}[$\sroles$]
We define the set of synchronising roles of a choreography $M$, $\sroles(M)$ by recursion on the structure of $M$: 
$\sroles(\com{S}{R})=\{S,R\}$, $\sroles(\select{S}{R}~l~M)=\{S,R\}\cup\sroles(M)$, $\sroles(f(\vec{R}))=\vec{R}$,
and homomorphically on all other cases.
\end{definition}  

\begin{example}\label{ex:rewritecom}
Consider the choreography $(\com SR~V@S)~((\fun{x}{T@R}M)~(\com SR~V'@S))$. Here, thanks to the restriction on synchronisation in \cref{rule:c-app3}, only the left $\com SR$ on $V$ is available. If we were to use a version of \cref{rule:r-absl} with no restriction on synchronisation roles to rewrite the choreography to $(\fun{x}{T@R}((\com SR~V@S)~M))~(\com SR~V'@S)$, we would instead have the rightmost $\com SR$ on $V'$ available. This means we have both communication available depending on whether we decide to rewrite and we have the same problem as in \cref{ex:comorder} of $S$ potentially choosing to send $V$ while $R$ has rewritten and is expecting to receive the left $\com{S}{R}$. We therefore do not allow such a rewrite and use synchronisation roles to prevent it.
\eoe
\end{example}

Our first result shows that closed choreographies remain closed under reductions.

\begin{proposition}
  \label{thm:Clo}
  Let $M$ be a closed choreography. If $M\rightarrow_D M'$ then $M'$ is closed.
\end{proposition}
\begin{proof}
  Straightforward from the semantics.
\end{proof}

One of the hallmark properties of choreographies is that well-typed choreographies should continue to reduce until they reach a value.
We split this result into two independent statements.

\begin{theorem}[Progress]\label{thm:TypeReduc}
  Let $M$ be a closed choreography and $D$ a collection of named choreographies with all the necessary definitions for $M$.
  If there exists a typing context $\Theta;\Sigma;\Gamma$ such that $\Theta;\Sigma;\Gamma\vdash M:T$ and $\Theta;\Sigma;\Gamma\vdash D$, then either $M$ is a value (and $M \centernot\rightarrow_{D}$) or there exists a choreography $M'$ such that $M\rightarrow_{D} M'$.
\end{theorem}
\begin{proof}
  Follows by induction on the typing derivation of $\Theta;\Sigma;\Gamma\vdash M:T$. See \cref{app:TypeReduc} for full details.
\end{proof}

\begin{theorem}[Type Preservation]\label{thm:TypePres}
  Let $M$ be a closed choreography.
  If there exists a typing context $\Theta;\Sigma;\Gamma$ such that $\Theta;\Sigma;\Gamma\vdash M:T$ and $\Theta;\Sigma;\Gamma\vdash D$, then $\Theta;\Sigma;\Gamma\vdash M':T$ for any $M'$ such that $M\rightarrow_D M'$.
\end{theorem}
\begin{proof}
  Follows from the typing and semantic rules. See \cref{app:TypeReduc} for full details.
\end{proof}

Combining these results, we conclude that if $M$ is a well-typed, closed, choreography, then either $M$ is a value or $M$ reduces to some well-typed, closed choreography $M'$.
Since $M'$ still satisfies the hypotheses of the above results, either it is a value or it can reduce.

\section{Illustrative Examples}
\label{sec:example}

In this section, we illustrate the expressivity of \chorlam with some representative examples. 
Specifically, we use \chorlam to implement the Diffie-Hellman protocol for key exchange~\cite{DH76}, a single sign-on authentication protocol, and the Extensible Authentication Protocol~\cite{eap}. 
The first two are used in \cite{GMP20} to illustrate the expressiveness of Choral, and we show how they can be adapted to \chorlam's functional paradigm.
The latter two examples require using higher-order composition of choreographies, as they are parametrised on respectively a protocol for creating channels (functions such as $\comsymbol$ which take a value at one role and returns the same value at a different role) and a list of authentication protocols. 
\chorlam is the first theory capable of modelling these choreographies as they are parametric on roles and include functions which are parametric on other choreographies and no previous formalism includes both these features.

\subsection{Secure Communication}\label{sec:seccom}

We write a choreography for the Diffie--Hellman key
exchange protocol~\cite{DH76}, which allows two roles to agree on a shared
secret key without assuming secrecy of communications. As in \cref{ex:choice}, we use the
primitive type $\typef{Int}$.

To define this protocol, we use the local function $\funf{modPow}(R)$ of the following
type
\[\funf{modPow}(R) :\IntAt{R}\to\IntAt{R}\to\IntAt{R}\to\IntAt{R}\]
which computes powers with a given modulo. Like all local functions in \chorlam, $\funf{modPow}(R)$ is modelled by a choreography located entirely at one role. Given $\funf{modPow}(R)$,
we can implement Diffie--Hellman as the following choreography:

\begin{lstlisting}[mathescape=true]
$\fundef{diffieHellman}(P,Q) = $
  $\fun{psk}{\IntAt{P}}\; \fun{qsk}{\IntAt{Q}}\; \fun{psg}{\IntAt{P}}\;
   \fun{qsg}{\IntAt{Q}}\; \fun{psp}{\IntAt{P}}\; \fun{qsp}{\IntAt{Q}}$
    $\pair \; (\funf{modPow}(P)\; psg \; (\com{Q}{P} \; (\funf{modPow}(Q)\; qsg \; qsk \; qsp)) \; psp)$
    $\hphantom{\pair \;} (\funf{modPow}(Q)\; qsg \; (\com{P}{Q} \; (\funf{modPow}(P)\; psg \; psk \; psp)) \; qsp)$
\end{lstlisting}

Given the individual secret keys ($psk$ and $qsk$) and a previously publicly
agreed upon shared prime modulus and base ($psg = qsg, psp = qsp$), the
participants exchange their locally-computed public keys in order to arrive at a
shared key that can be used to encrypt all further communication. The
choreography has the type
\[\IntAt{P}\to\IntAt{Q}\to\IntAt{P}\to\IntAt{Q}\to\IntAt{P}\to\IntAt{Q}\to\pairt{\IntAt{P}}{\IntAt{Q}}\]
and represents the shared key as a pair of equal keys, one for each participant.

Using the key exchange protocol, we can now build a reusable utility that allows
us to achieve secure bidirectional communication between the parties, by
encrypting and decrypting messages with the shared key at the appropriate
endpoints. For this we assume two functions that allow us to encrypt and decrypt
a $\typef{String}$ message with a given $\typef{Int}$ key:
\begin{align*}
  \funf{encrypt}(R) &: \IntAt{R}\to\StringAt{R}\to\StringAt{R} \\
  \funf{decrypt}(R) &: \IntAt{R}\to\StringAt{R}\to\StringAt{R}
\end{align*}

The choreography then takes a shared key as its parameter and produces a pair of
unidirectional channels that wrap the communication primitive with the necessary
encryption based on the key:

\begin{lstlisting}[mathescape=true]
$\fundef{makeSecureChannels}(P, Q) = \fun{key}{\pairt{\IntAt{P}}{\IntAt{Q}}}$
  $\Pair\;(\fun{val}{\StringAt{P}}\;(\funf{decrypt}(Q)\; (\snd\;key)\; (\com{P}{Q}\;(\funf{encrypt}(P)\; (\fst\;key)\; val))))$
  $\hphantom{\Pair}\;(\fun{val}{\StringAt{Q}}\;(\funf{decrypt}(P)\; (\fst\;key)\; (\com{Q}{P}\;(\funf{encrypt}(Q)\; (\snd\;key)\; val))))$
\end{lstlisting}

The fact that this choreography returns a pair of channels can also be seen from
its type:
\[(\pairt{\IntAt{P}}{\IntAt{Q}})\to(\pairt{(\StringAt{P}\to \StringAt{Q})}{(\StringAt{Q}\to \StringAt{P})})\]

Using the channels is as easy as using $\boldsymbol{\mathsf{com}}$ itself and
amounts to a function application.

\subsection{Single Sign-on Authentication}
\label{sec:auth}

We now implement the single sign-on authentication protocol inspired by the OpenID
specification~\cite{openid} described in \cref{ex:authchor}. We first implement the choreography in a parametric
way that allows us to specify the means of communication, and then combine it
with the secure communication from the previous example.

The protocol involves three roles with the goal being for the client $C$ to
authenticate with the server $S$ via a third party identity provider $I$. If
authentication succeeds, the client and the server both get a unique token from
the identity provider.

We use the following local functions for working with user credentials
\begin{align*}
  \funf{username}(R) &: \typefAt{Credentials}{R}\to\StringAt{R} \\
  \funf{password}(R) &: \typefAt{Credentials}{R}\to\StringAt{R} \\
  \funf{calcHash}(R) &: \StringAt{R}\to\StringAt{R}\to\StringAt{R}
\end{align*}
computing, respectively, the username and password from a local type
$\typefAt{Credentials}{R}$ (which can be a pair, for example),
and the hash of a string with a given salt. These are mainly used by the client.

In addition, we require functions for retrieving the salt, validating the hash,
and creating a token for a given username, which are used by the identity
provider:
\begin{align*}
  \funf{getSalt}(R) &: \StringAt{R}\to\StringAt{R} \\
  \funf{check}(R) &: \StringAt{R}\to\StringAt{R}\to\BoolAt{R} \\
  \funf{createToken}(R) &: \StringAt{R}\to\StringAt{R}
  .
\end{align*}

Given the above helper functions, the authentication protocol is as follows.
Here we use if-then-else as syntactic sugar for $\ccase$:

\begin{lstlisting}[mathescape=true]
$\fundef{authenticate}(S, C, I) = \fun{credentials}{\typefAt{Credentials}{C}}$
  $\fun{comcip}{\StringAt{C}\to \StringAt{I}}\;\fun{comipc}{\StringAt{I}\to \StringAt{C}}$
  $\fun{comips}{\StringAt{I}\to \StringAt{S}}$
    $((\fun{user}{\StringAt{I}}\;(\fun{salt}{\StringAt{C}}\;(\fun{hash}{\StringAt{I}}$
      $\cif\;\funf{check}(I)\; user \; hash \; \cthen$
        $\select{I}{C}\;\lab{ok} \; (\select{I}{S}\;\lab{ok}$
          $(\fun{token}{\StringAt{I}}\;\inl\;(\pair\;(comipc\; token)\;(comips\; token)))$ 
          $(\funf{createToken}(I)\; user))$
      $\celse$
        $\select{I}{C}\;\lab{ko} \; (\select{I}{S}\;\lab{ko} \; \inr\;
        \litAt{\unit}{I}))$
    $\;\;\;\;\;(comcip\; (\funf{calcHash}(C)\; salt \; (\funf{password}(C)\; credentials))))$
    $\;\;\;(comipc\; (\funf{getSalt}(I)\; user)))$
    $\;(comcip\; (\funf{username}(C)\; credentials)))$
\end{lstlisting}

As mentioned, the choreography is parametrised over three channels between the
participants, allowing the communication to be customized ($comcip$, $comipc$
and $comips$). The client first sends their username to the identity provider
who replies with the appropriate salt. The client then calculates a salted hash
of their password and sends it back to the identity provider. Finally, the
identity provider validates the hash and either sends a token to both
participants or returns a unit. The shared token is again represented using a
pair of equal values, visible from the type of the choreography:
\begin{multline*}
  \typefAt{Credentials}{C}\to(\StringAt{C}\to\StringAt{I})\to(\StringAt{I}\to\StringAt{C})\\
  \to(\StringAt{I}\to\StringAt{S})\to((\pairt{\StringAt{C}}{\StringAt{S}})+\UnitAt{I})
\end{multline*}

By combining the choreographies $\funf{authenticate}(S, C, I)$ and
$\funf{makeSecureChannels}(P, Q)$ (from \cref{sec:seccom}), we can obtain a choreography $\funf{main}(S, C, I)$ that carries out the
authentication securely. Using $\funf{makeSecureChannels}(P, Q)$, the
participants first establish secure channels backed by encryption
keys derived using $\funf{diffieHellman}(P, Q)$. After the secure communication
is in place, the participants can execute the authentication protocol specified
by $\funf{authenticate}(S, C, I)$.

\begin{lstlisting}[mathescape=true]
$\fundef{main}(S, C, I) =$
  $(\fun{k1}{\pairt{\IntAt{C}}{\IntAt{I}}} \;
  \fun{k2}{\pairt{\IntAt{I}}{\IntAt{S}}}$
    $(\fun{c1}{\pairt{(\StringAt{C} \to \StringAt{I})}{(\StringAt{I} \to \StringAt{C})}}$
    $\;\fun{c2}{\pairt{(\StringAt{I} \to \StringAt{S})}{(\StringAt{S} \to \StringAt{I})}}$
      $(\fun{t}{(\pairt{\StringAt{C}}{\StringAt{S}})+\UnitAt{I}}\;$
        $\ccase\;t\;\cof$
          $\cleft{x} \litAt{\text{"Authentication successful"}}{C}$
          $\cright{x} \litAt{\text{"Authentication failed"}}{C})$
      $(\funf{authenticate}(S, C, I)\; (\fst\;c1) \; (\snd\;c1) \; (\fst\;c2)))$
    $(\funf{makeSecureChannels}(C, I)\; k1)
      \;(\funf{makeSecureChannels}(I, S)\; k2))$
  $(\funf{diffieHellman}(C, I)\; csk \; ipsk \; csg \; ipsg \; csp \; ipsp)$
  $(\funf{diffieHellman}(I, S)\; ipsk \; ssk \; ipsg \; ssg \; ipsp \; ssp)$
\end{lstlisting}

In this example, the client simply reports whether the authentication has
succeeded with a value, which can be checked in a larger context. Or, alternatively, we could parameterise $\mathsf{main}$ over choreographic continuations to be invoked in case of success or failure.

We denote by $\Gamma$ the set of typings we have given so far in this section.
Then we can type
\[
\{S,C,I\};\emptyset;\Gamma\vdash\funf{main}(S,C,I):\StringAt{C}
.
\]

\subsection{EAP}\label{sec:EAP}

Finally, we turn to implementing the core of the Extensible Authentication
Protocol (EAP)~\cite{eap}. EAP is a widely-employed link-layer protocol that
defines an authentication framework allowing a peer $P$ to authenticate with a
backend authentication server $S$, with the communication passing through an
authenticator $A$ that acts as an access point for the network.

The framework provides a core protocol parametrised over a set of authentication
methods (either predefined or custom vendor-specific ones), modelled as
individual choreographies  with type $\typefAt{AuthMethod}{(P,A,S)} = \StringAt{S} \to_{\{P, A\}} \BoolAt{S}\text.$
For this reason, it is desirable that the core of the
protocol be written in a way that doesn't assume any particular authentication
method.

The $\funf{eap}(P, A, S)$ choreography does exactly that by leveraging
higher-order composition of choreographies:
\begin{lstlisting}[mathescape=true]
$\fundef{eap}(P, A, S) = \fun{methods}{\ListAt{\typef{AuthMethod}}{(P, A, S)}}$
  $\funf{eapAuth}(P, A, S) \; (\funf{eapIdentity} \; \litAt{\text{"auth request"}}{S}) \; methods$

$\fundef{eapAuth}(P, A, S) = \fun{id}{\StringAt{S}} \; \fun{methods}{\ListAt{\typef{AuthMethod}}{(P, A, S)}}$
  $\cif \; \funf{empty}(P, A, S) \; methods \; \cthen$
    $\funf{eapFailure}(P, A, S) \; \litAt{\text{"try again later"}}{S}$
  $\celse$
    $\cif \; (\fst \; methods) \; id \; \cthen$
      $\select{S}{P}\;\lab{ok}\;(\select{S}{A}\;\lab{ok}\;(\funf{eapSuccess}(P, A, S) \; \litAt{\text{"welcome"}}{S}))$
    $\celse$
      $\select{S}{P}\;\lab{ko}\;(\select{S}{A}\;\lab{ko}\;(\funf{eapAuth}(P, A, S) \; id \; (\snd \; methods)))$
\end{lstlisting}

For the sake of simplicity, we have left out the definitions of a couple of helper
choreographies that are referenced in the example:
\begin{align*}
  \funf{eapIdentity}(P, A, S) &: \StringAt{S}\to_{\{P, A\}}\StringAt{S} \\
  \funf{eapSuccess}(P, A, S) &: \StringAt{S}\to(\StringAt{P} \times \StringAt{A}) \\
  \funf{eapFailure}(P, A, S) &: \StringAt{S}\to(\StringAt{P}\times\StringAt{A})
\end{align*}

The $\funf{eap}(P, A, S)$ choreography first fetches the client's identity using
$\funf{eapIdentity}(P, A, S)$ which exchanges the proper EAP packets and
delivers the client's identity to the server.

Once the identity is known, $\funf{eapAuth}(P, A, S)$ is invoked in order to try
the list of authentication methods in order until one succeeds, or the list is
exhausted and authentication fails. Just like the previous example, EAP is parametric on a choreography, or in this case a list of choreographies, $methods$. 
We use the notation for lists in $\ListAt{\typef{AuthMethod}}{(P, A, S)}$ as described in \cref{ex:choice}, while the function $\funf{empty}(P, A, S)$ allows us to determine whether the list of methods is empty.
Note that each authentication method can be an arbitrarily-complex choreography with its own communication structures that can involve all three involved roles, and implements a particular authentication method on top of EAP.

Finally, depending on the outcome of the authentication, an appropriate EAP
packet is sent with either $\funf{eapSuccess}(P, A, S)$ or $\funf{eapFailure}(P, A, S)$ to indicate the result to the client.

\section{Endpoint Projection}
\label{sec:epp}

Now that we have seen the kind of protocols we can describe with \chorlam, we want to define a distributed implementation of choreographies at the involved roles.
In order to implement a choreography, one must determine how each individual role behaves.
We introduce a distributed $\lambda$-calculus to specify these behaviours, and show how to generate
implementations of choreographies automatically.
In this context, roles are implemented by ($\lambda$-calculus) \emph{processes} and we use the two terms interchangeably. Our calculus resembles the distributed $\pi$-calculus \cite{H07} in that it puts $\lambda$-calculus processes at roles.

\subsection{Process Language}
\begin{definition}
  The syntax of behaviours, local values and local types is defined by the following grammar.
  \begin{align*}
    B &\Coloneqq L \mid f(\vec{\role R}) \mid B~B \mid \case{B}{x}{B}{x}{B} \mid \oplus_{\role R}~\ell~B \\
    &\hphantom{\Coloneqq} \mid \&_{\role R}\{\ell_1:B_1,\dots, \ell_n:B_n\}\\
    L &\Coloneqq x \mid \lambda x:T.B \mid \Inl~L \mid \Inr~L \mid \fst \mid \snd \mid \Pair~L~L \mid \unit \mid \recv{\role R} \mid \send{\role R} \mid \botm &\\
    T &\Coloneqq T \to T \mid T + T \mid T\times T \mid \unitt \mid t_i \mid \botmt
  \end{align*}
\end{definition}
Behaviours correspond directly to local counterparts of choreographic actions.
The terms from the $\lambda$-calculus, as well as pairs and sums, are unchanged (except that there are no role annotations now);
the choreographic actions generate two terms each.
Selection yields the \emph{offer} term $\&_{\role R}\{\ell_1:B_1,\dots, \ell_n:B_n\}$, which offers a number
of different ways it can continue for another process $\role{R}$ to choose from, and the \emph{choice} term
$\oplus_R~\ell~B$, which directs $\role{R}$ to continue as the behaviour labelled $\ell$.
Likewise, communication has been divided into a \emph{send} to $\role{R}$ action, $\send{R}$, and a
\emph{receive} from $\role{R}$ action, $\recv{R}$.
We also add a $\botm$ value and type, which we use when projecting choreographies or types onto roles where thy are not present, e.g. $\unit@R$ projected to $\roles{S}$ will be $\botm$ and also have the type $\botmt$.
Types are otherwise defined exactly as for choreographies, but without roles. The local type variable $t_i$ is the local type of $R_i$ in a choreography with type $t@(R_1,\dots,R_n)$.

\begin{definition}
  A network $\mathcal N$ is a finite map from a set of processes to behaviours.
\end{definition}

The parallel composition of two networks $\mathcal N$ and $\mathcal N'$ with disjoint domains,
$\mathcal N\mid\mathcal N'$, simply assigns to each process its behaviour in the network defining it.
Any network is equivalent to a parallel composition of networks with singleton domain, and therefore
we often write $R_1[B_1]\mid\ldots\mid R_n[B_n]$ for the network where process $R_i$ has behaviour
$B_i$.

\begin{figure}
	\begin{spreadlines}{\rulesvsep}
  \begin{eqnarray*}&
     \infer[{\rtag{\rlabel{\rname{NSend}}{rule:n-send}}}]{\send{\role R}~L\xrightarrow{\send{\role R}~L}_{\DP} \botm}{\fv(L)=\emptyset}
    \qquad
    \infer[{\rtag{\rlabel{\rname{NRecv}}{rule:n-recv}}}]
     {\recv{\role R}~\botm\xrightarrow{\recv{\role R}~L}_{\DP} L}
     {}
    \\&
     \infer[{\rtag{\rlabel{\rname{NCom}}{rule:n-com}}}]
     {S[B_1] \mid R[B_2]\xrightarrow{\tau_{\role S,\role R}}_{\DN} \role S[B'_1] \mid \role R[B'_2]}
     {B\xrightarrow{\send{\role S}~L}_{\DN(S)} B'_1 
     & B_2\xrightarrow{\recv{\role R}~L[\role S:=\role R]}_{\DN(R)} B'_2}
   \\&
   \infer[{\rtag{\rlabel{\rname{NCho}}{rule:n-cho}}}]
     {\oplus_{\role R}~l~B\xrightarrow{\oplus_{\role R}~l}_{\DP} B}{}
    \qquad
    \infer[{\rtag{\rlabel{\rname{NOff}}{rule:n-off}}}]
    {\&_{\role R}\{\ell_1:B_1,\dots,\ell_n:B_n\}\xrightarrow{\&_{\role R}\ell_i}_{\DP} B_i}{}
     \\& 
     \infer[{\rtag{\rlabel{\rname{NOff2}}{rule:n-off2}}}]{\&_{\role R}\{\ell_1:B_1,\dots,\ell_n:B_n\}\xrightarrow{\mu}_{\DP} \&_{\role R}\{\ell_1:B_1',\dots,\ell_n:B_n'\}}{B_i\xrightarrow{\mu}_{\DP}B_i'\text{ for }1\leq i\leq n & \mu\in\{\tau,\lambda\}}
    \\&
    \infer[{\rtag{\rlabel{\rname{NCho2}}{rule:n-cho2}}}]{\oplus_{\role R}~l~B\xrightarrow{\mu}_{\DP} \oplus_{\role R}~l~B'}{B\xrightarrow{\mu}_{\DP} B' & \mu\in\{\tau,\lambda\}}
     \qquad 
     \infer[{\rtag{\rlabel{\rname{NSel}}{rule:n-sel}}}]
     	{\role S[B_1] \mid R[B_2]\xrightarrow{\tau_{\role S,\role R}}_{\DN} \role S[B'_1] \mid \role R[B'_2]}
     	{B_1\xrightarrow{\oplus_{\role R}~\ell}_{\DN(S)} B'_1 
     	& B_2\xrightarrow{\&_{\role S}~\ell}_{\DN(R)} B'_2} 
    \\&
    \infer[{\rtag{\rlabel{\rname{NAbsApp}}{rule:n-absapp}}}]
     {(\lambda x:T.B)~L\xrightarrow{\tau}_\DP B[x:=L]}{}
     \qquad 
     \infer[\rtag{\rlabel{\rname{NInAbs}}{rule:n-inabs}}]
     	{\lambda x:T.B\xrightarrow{\lambda}_D \lambda x:T.B'}
     	{B\xrightarrow{\mu}_\DP B' & \mu\in\{\tau,\lambda\} }
    \\&
     \infer[{\rtag{\rlabel{\rname{NApp1}}{rule:n-app1}}}]
     	{B~B'\xrightarrow{\mu'}_\DP B''~B'}
     	{B\xrightarrow{\mu}_\DP B'' & \text{if }\mu=\lambda \text{ then } \mu'=\tau \text{ else } \mu'=\mu}
    \\&
     \infer[{\rtag{\rlabel{\rname{NApp2}}{rule:n-app2}}}]
     	{L~B\xrightarrow{\mu}_\DP L~B'}
     	{B\xrightarrow{\mu}_\DP B'}  
     	\qquad
     \infer[{\rtag{\rlabel{\rname{NApp2}}{rule:n-app3}}}]
     	{B~B'\xrightarrow{\tau}_\DP B~B''}
     	{B'\xrightarrow{\tau}_\DP B''}  \\&
     	\infer[{\rtag{\rlabel{\rname{NCase}}{rule:n-case}}}]
     	{\case{B}{x}{B'}{x'}{B''}\xrightarrow{\mu}_\DP \case{B'''}{x}{B'}{x'}{B''}}
     	{B\xrightarrow{\mu}_\DP B'''} \\&
     	\infer[{\rtag{\rlabel{\rname{NCase2}}{rule:n-case2}}}]
     	{\case{B}{x}{B_1}{x'}{B_2}\xrightarrow{\mu}_\DP \case{B}{x}{B_1'}{x'}{B_2'}}
     	{B_1\xrightarrow{\mu}_\DP B_1' & B_2\xrightarrow{\mu}_\DP B_2' & \mu\in\{\lambda,\tau\}}  
    \\&
     \infer[{\rtag{\rlabel{\rname{NPro}}{rule:n-pro}}}]
     	{\role R[B]\xrightarrow{\tau_R}_{\DN} R[B']}
     	{B\xrightarrow{\tau}_{\DN(R)} B'}
    \qquad
     	\infer[{\rtag{\rlabel{\rname{NPar}}{rule:n-par}}}]
     		{\mathcal{N}\mid\mathcal{N}'\xrightarrow{\tau_{\mathbf{R}}}_{\DN} \mathcal{N}''\mid\mathcal{N}'}
     		{\mathcal{N}\xrightarrow{\tau_{\mathbf{R}}}_{\DN} \mathcal{N}''} \\&	
     \infer[\rtag{\rlabel{\rname{NStr}}{rule:n-str}}]{B\xrightarrow{\mu}_\DP B'}{B\rightsquigarrow^* B'' & B''\xrightarrow{\mu} B'}
  \end{eqnarray*}
  \end{spreadlines}
  \caption{Network semantics (representative rules).}
  \label{fig:network-semantics}
\end{figure}
\begin{figure}[t]
\begin{spreadlines}{\rulesvsep}
  \begin{eqnarray*}&
  \infer[{\rtag{\rlabel{\rname{LR-AbsR}}{rule:lr-absr}}}]
  {((\lambda x.B)~B')~B''\rightsquigarrow (\lambda x.B~B'')~B')} 
  {}
  \qquad
\infer[{\rtag{\rlabel{\rname{LR-AbsL}}{rule:lr-absl}}}]{B''~((\lambda x.B)~B')\rightsquigarrow (\lambda x.B''~B)~B')}{\roles(B'')=\emptyset}  
&\\&
\infer[{\rtag{\rlabel{\rname{LR-CaseR}}{rule:lr-caser}}}]
{\begin{array}{l}(\case{B}{x}{B_1}{x}{B_2})~B'\rightsquigarrow \\ \qquad\qquad\qquad\qquad \case{B}{x}{(B_1~B')}{x}{(B_2~B')}\end{array}}
{}
&\\&
\infer[{\rtag{\rlabel{\rname{LR-CaseL}}{rule:lr-caseL}}}]{\begin{array}{l}B'~(\case{B}{x}{B_1}{x}{B_2})\rightsquigarrow \\ \qquad\qquad\qquad\qquad \case{B}{x}{(B'~B_1)}{x}{(B'~B_2)}\end{array}}{\roles(B')=\emptyset}
&\\&
\infer[{\rtag{\rlabel{\rname{LR-OffL}}{rule:lr-offl}}}]{B~(\&_{\role R}\{\ell_1:B_1,\dots,\ell_n:B_n\})\rightsquigarrow \&_{\role R}\{\ell_1:B~B_1,\dots,\ell_n:B~B_n\}}{\roles(B)=\emptyset}
&\\&
\infer[{\rtag{\rlabel{\rname{LR-OffR}}{rule:lr-offr}}}]
{(\&_{\role R}\{\ell_1:B_1,\dots,\ell_n:B_n\})~B\rightsquigarrow \&_{\role R}\{\ell_1:B_1~B,\dots,\ell_n:B_n~B\}}
{}
&\\&
\infer[{\rtag{\rlabel{\rname{LR-ChoL}}{rule:lr-chol}}}]{B'~(\oplus_{\role R}~l~B)\rightsquigarrow \oplus_{\role R}~l~(B'~B)}{\roles(B')=\emptyset}
\qquad
\infer[{\rtag{\rlabel{\rname{LR-ChoR}}{rule:lr-chor}}}]
{(\oplus_{\role R}~l~B)~B'\rightsquigarrow \oplus_{\role R}~l~(B~B')}
{}
&\\&
\infer[{\rtag{\rlabel{\rname{LR-Botm}}{rule:lr-botm}}}]
{\botm~\botm\rightsquigarrow \botm}
{}
\end{eqnarray*}
  \end{spreadlines}
  \caption{Rewriting of processes.}
  \label{fig:rewriteP}
\end{figure}

The semantics of networks is given as a labelled transition system.
Representative rules that define transitions are included in \cref{fig:network-semantics}.

Most of these rules are similar to the ones for choreographies; the difference is that communications and selections now require synchronisation between the processes implementing the two local actions.
This is achieved by matching the appropriate labels on the reductions.

Our network semantics are unusual because out-of-order execution is allowed not only at the choreography level, as is common, but also at the process level. This is necessary because an action in the choreography can correspond to an internal action at multiple roles.

\begin{example}
Consider the choreography \begin{align*}M={}&\ccase\;N@S\;\cof\;\cleft{x}{((\lambda y:T@S\rightarrow_{\emptyset} T@R.y~x)~\com{S}{R})}\\&
\hphantom{\ccase\;N@S\;\cof\;}\cright{x'}{((\lambda y:T@S\rightarrow_{\emptyset} T@R.y~x')~\com{S}{R})}
\end{align*}
keeping in mind that $\com{S}{R}$ is a value and can be used as such in an application. Since the conditional and $N$ are located at $S$, the part of $M$ located at $R$ is an application of a receive, $(\lambda y:\botmt\rightarrow T.y~\botm)~\recv{S}$, which lets us perform the application resulting in $\recv{S}~\botm$. Since we have out-of-order-execution rules in our choreography semantics, we can perform the corresponding action in the choreography 
\[M\xrightarrow{\tau,\emptyset} \case{N@S}{x}{\com{S}{R}~x}{x'}{\com{S}{R}~x'}\]
Since $\com{S}{R}$ is located at both $S$ (as $\send{R}$) and $R$ (as $\recv{R}$), this action at $M$ also corresponds to an application at $S$. However, at $S$ this application is guarded by the conditional,with $M$ at $S$ corresponding to the process $\case{N}{x}{((\lambda y:T\rightarrow\botmt.y~x)~\send{R})}{x'}{((\lambda y:T\rightarrow\botmt.y~x')~\send{R})}$. We therefore need to allow out-of-order execution at the processes as well as the choreographies.

Note that since $R$ must wait for $S$ to be ready to synchronise before performing the receive, we do not allow out-of-order execution of communications or selections.
\eoe
\end{example}

In addition, we have another rewriting rule, \cref{rule:lr-botm}, for collecting $\botm$ processes. The need for this rule is illustrated by \cref{ex:botm}.

\begin{example}\label{ex:botm}
Consider the choreography $(\com{S}{R}~(\lambda x:T@S.M@S))~(\com{S}{R}~V@S)$, where a function and a value are both sent from $\role{S}$ to $\role{R}$ before being applied at $\role{R}$. At $\role{R}$ this is the process $(\recv{S}~\botm)~(\recv{S}~\botm)$, which executes as we want. However, at $\roles{S}$ this choreography becomes the process $(\send{R} (\lambda x:T.M))~(\send{R}~V)$, which after the two communications are executed becomes $\botm~\botm$. After the choreography has executed both communications it is $(\lambda x:T@R.M@R)~V@R$, which is located entirely at $R$, and therefore at $S$ becomes $\botm$. We therefore need a way to make the application $\botm~\botm$ in a process become only one $\botm$. The rewriting rule \cref{rule:lr-botm} serves this purpose.
\eoe
\end{example}

\begin{figure}[t]
	\begin{spreadlines}{\rulesvsep}
  \begin{eqnarray*}&
     \infer[{\rtag{\rlabel{\rname{NTChor}}{rule:nt-cho}}}]
     	{\Sigma;\Gamma\vdash \oplus_{\role R}~\ell~B:T}
     	{ \Sigma;\Gamma\vdash B:T}
    \qquad
     \infer[{\rtag{\rlabel{\rname{NTOff}}{rule:nt-off}}}]
     	{\Sigma;\Gamma\vdash \&_{\role R}\{\ell_1:B_1,\dots\ell_n:B_n\}:T}
     	{\Sigma;\Gamma\vdash B_i:T \text{ for }1\leq i\leq n} 
    \\&
      \infer[{\rtag{\rlabel{\rname{NTSend}}{rule:nt-send}}}]
     {\Sigma;\Gamma\vdash \send{\role R}:T \rightarrow \botmt}
     {}
    \qquad
    \infer[{\rtag{\rlabel{\rname{NTRecv}}{rule:nt-recv}}}]
     {\Sigma;\Gamma\vdash \recv{\role R}:\botmt \rightarrow T}
     {}
     \\&
     \infer[{\rtag{\rlabel{\rname{NTbotm}}{rule:nt-botm}}}]
     	{\Sigma;\Gamma\vdash \botm:\botmt}
       {}
     	\qquad
     \infer[{\rtag{\rlabel{\rname{NTApp2}}{rule:nt-app2}}}]
     	{\Sigma;\Gamma\vdash B~B':\botmt}
     	{ \Sigma;\Gamma\vdash B:\botmt & \Sigma;\Gamma\vdash B':\botmt}
  \end{eqnarray*}
  \end{spreadlines}
  \caption{Typing rules for behaviours (representative rules).}
  \label{fig:network-typing}
\end{figure}
Our calculus includes a typing system for typing behaviours.
Typing judgements now have the form $\Sigma;\Gamma\vdash B:T$.
Most of the rules are direct counterparts to those in \cref{fig:typing}, obtained by removing $\Theta$ and any side conditions involving roles. We also add the $\botmt$ type, used for typing $\botm$ values or choreographies resulting $\botm$ values, including applications of two such choreographies, which can always be reduced to $\botm$ using \cref{rule:lr-botm}. 
A representative section of the new rules are given in \cref{fig:network-typing}.%

\subsection{Endpoint Projection (EPP)}
We now have the necessary ingredients to define the endpoint projection (EPP) of a choreography $M$ for an individual role $\role{R}$ given a typing derivation showing that $\Theta;\Sigma;\Gamma\vdash M:T$ for some type $T$.
Formally, the definition of EPP depends on this typing derivation; but to keep notation simple we write $\epp MR$ and refer to the type $T$ associated to $M$ in the derivation as $\type(T)$.

Intuitively, the projection simply translates each choreography action to the corresponding local behaviour.
For example, a communication action projects to a send (for the sender), a receive (for the receiver), or a unit (for the remaining processes).
In order to define EPP precisely, we need a few additional ingredients, which we briefly describe.

Projecting a term requires knowing the roles involved in its type.
This is implicitly given in the derivation provided to EPP.
It can easily be shown by structural induction that, if the derivation contains two different typing judgements for the same term, then the roles involved in that term's type are the same.
So we write without ambiguity $\roles(\type(M))$ for this set of roles.

The second ingredient concerns knowledge of choice.
When projecting $\case Mx{M'}y{M''}$, roles not occurring in $M$ cannot know what branch of the choreography is chosen; therefore, the projections of $M'$ and $M''$ must be combined in a uniquely defined behaviour.
This is done by means of a standard partial \emph{merge} operator ($\merge$), adapted from~\cite{CHY12,CM20,HYC16}, whose key property is
\[\&\{\ell_i:B_i\}_{i\in I} \merge \&\{\ell_j:B'_j\}_{j\in J}= \&\left(\{\ell_k:B_k\merge B'_k\}_{k\in I\cap J} \cup \{\ell_i:B_i\}_{i\in I\setminus J}\cup \{\ell_j:B'_j\}_{j\in J\setminus I}\right)\]
and which is homomorphically defined for the remaining constructs (see the \cref{sec:appendix} for the full definition).
Merging of incompatible behaviours is undefined.

\begin{definition}
  The EPP of a choreography $M$ for role $\role{R}$ is defined by the rules in \cref{fig:projection}.

  To project a network from a choreography, we therefore project the choreography for each role and combine the results in parallel:
  $\epp M{}=\prod_{R\in \roles(M)} R\left[\epp MR\right]$.
\end{definition}

\begin{figure}[p]
\begin{small}
  \begin{eqnarray*}
    \lefteqn{\mbox{Choreographies:}} \\
    && \epp{M~N}R =
    \begin{dcases*}\epp MR~\epp NR & if $R\in \roles(\type(M))$ or $R\in\roles(M)\cap\roles(N)$ \\
    \botm & if $\epp MR=\epp NR=\botm$ \\ 
    \epp MR & if $\epp NR=\botm$ \\
    \epp NR & otherwise \\
\end{dcases*} \\[5pt]
    && \epp{\lambda x:T.M}R =
    \begin{dcases*} \lambda x:\epp{T}{R}.\epp MR & if $R\in \roles(\type(x:T.M))$ \\
      \botm & otherwise \end{dcases*} \\[5pt]
    && \epp{\case{M}{x}{N}{x'}{N'}}R = \\[5pt]
    && \qquad
    \begin{dcases*} \case{\epp MR}{x}{\epp NR}{x'}{\epp{N'}R} & if $R\in \roles(\type(M))$ \\
    \botm & if $\epp MR= \epp NR= \epp {N'}R=\botm$ \\
    \epp MR & if $\epp NR=\epp {N'}R=\botm$ \\
    \epp NR \merge \epp {N'}R & if $\epp MR=\botm$ \\
      (\lambda x'':\botmt.\epp NR \merge \epp{N'}R)~\epp MR & for some $x''\notin {\fv(N) \cup \fv(N')}$ \\ & otherwise \end{dcases*} \\[5pt]
    && \epp{\select{S}{S'}~\ell~M}R =
    \begin{dcases*}\oplus_{S'}~\ell~\epp MR & if $R=S\neq S'$ \\
      \&_S\{\ell:\epp MR\} & if $R=S'\neq S$ \\
      \epp MR & otherwise \end{dcases*} \\[5pt]
    && \epp{\com{S}{S'}}R =
    \begin{dcases*} \lambda x:\epp{T}{R}.x & if $R=S=S'$ and $\type(\com{S}{S'})= T\rightarrow_{\emptyset} T'$\\
      \send{S'} & if $R=S\neq S'$\\
      \recv{S} & if $R=S'\neq S$\\
      \botm & otherwise \end{dcases*} \\
    && \epp{\litAt{\unit}{S}}R= \begin{dcases*} \unit & if $S=R$ \\
    \botm & otherwise
    \end{dcases*} 
    \qquad\qquad
    \epp xR =
    \begin{dcases*} x & if $R\in \roles(\type(x))$ \\ \botm & otherwise \end{dcases*}
    \\
    &&
    \epp {f(\vec{R})}R=
    \begin{dcases*}f_i( R_1,\dots,R_{i-1},R_{i+1},\dots, R_n) & if $\vec{R}= R_1,\dots,R_{i-1},R,R_{i+1},\dots, R_n$ \\
    \botm & otherwise
    \end{dcases*}
    \\[5pt]
    \lefteqn{\mbox{Types:}} \\
    && \epp{\unitt@S}R = \begin{dcases*} \unitt & if $S=R$ \\
    \botmt & otherwise
    \end{dcases*} 
    \qquad
    \epp{T\times T'}R =
    \begin{dcases*} \epp TR \times \epp{T'}R & if $R\in\roles(T\times T')$ \\
      \botmt & otherwise \end{dcases*} \\[5pt]
    && \epp{t@\vec{R}}R =
    \begin{dcases*} t_i & if $\vec{R}= R_1,\dots,R_{i-1},R,R_{i+1},\dots, R_n$ \\ \botmt & otherwise \end{dcases*}
    \\[5pt]
    && \epp{T\rightarrow_\rho T'}R =
    \begin{dcases*}\epp TR \rightarrow \epp{T'}R & if $R\in \rho\cup\roles(T)\cup\roles(T')$ \\
      \botmt & otherwise \end{dcases*}  \\[5pt]
    \lefteqn{\mbox{Definitions:}} \\
    && \epp D{} = \{f_i( R_1,\dots, R_{i-1},R_{i+1},\dots, R_n) \mapsto \epp{D(f( R_1,\dots, R_n))}{R_i} \mid f( R_1,\dots, R_n)\in\mathsf{domain}(D) \}
  \end{eqnarray*}
  \end{small}
  \caption{Projecting a choreography in \chorlam onto a role}
  \label{fig:projection}
\end{figure}

Intuitively, projecting a choreography to a role that is not involved in it returns a $\botm$.
More complex choreographies, though, may involve roles that are not shown in their type.
This explains the first clause for projecting an application: even if $R$ does not appear in the type of $M$, it may participate in interactions inside $M$.
A similar observation applies to the projection of $\ccase$, where merging is also used.

Selections and communications follow the intuition given above, with one interesting detail:
self-selections are ignored, and self-communications project to the identity function.
This is different from many standard choreography calculi, where self-communications are not allowed---we do not want to impose this in \chorlam, since one of the planned future developments for this language is to add polymorphism.

Likewise, projecting a type yields $\botmt$ at any role not used in that type.
The projection of a set of function definitions maps choreography names to behaviours.

\begin{proposition}
  Let $M$ be a closed choreography.
  If $\Theta;\Sigma;\Gamma\vdash M:T$, then for any role $\role{R}$ appearing in $M$,
  we have that $\epp\Sigma{};\epp\Gamma{}\vdash\epp MR:\epp TR$, where $\epp\Sigma{}$ and
  $\epp\Gamma{}$ are defined by applying EPP to all the types occurring in those sets.
\end{proposition}
\begin{proof}
  Straightforward from the typing and projection rules.
\end{proof}

\begin{example}
  The projections of the choreographies in \cref{ex:choice} are the following.
\begin{lstlisting}[mathescape=true]
$\epp{D(\funf{remoteFunction}(R, S))}1(S) = \fun{f}{(\typef{Int}\to
\typef{Int})}\;\fun{val}{\botmt}\;\send{S}\;(f\;(\recv{S}\;\botmt))$
$\epp{D(\funf{remoteFunction}(R, S))}2(R) =
\fun{f}{\botmt}\;\fun{val}{\typef{Int}}\;\recv{R}\;(\send{R}\;val)$
\end{lstlisting}
  This example illustrates the key features discussed in the text: projection of communications as
  two dual actions; and the way
  function applications are projected when the role does not appear in the function's type.
  \eoe
\end{example}

We describe what we consider adequacy of choreographic projection, formally defined in \cref{def:adeq}. Adequacy means that for any context $C[]$ the projection of $C[M]$ at $R$ will be the same for any $M$ which does not involve $R$. The definition of contest is as expected and can be found in \cref{sec:appendix}. Adequacy is usual (and expected) in choreographic programming, because the projection of $R$ should not be generating junk code based on the behaviour of other roles.

\begin{definition}[Adequacy]\label{def:adeq}
An EPP is adequate if $\epp{C[M]}{R}=\epp{C[N]}{R}$ for any role $R$, context $C[]$, and choreographies $M$ and $N$ such that $R\notin \roles(M,N,\type(M),\type(N))$.
\end{definition}

Adequacy ensures that if we modify a part of the choreography in which $R$ is not involved, we do not need to recompile the projection of $R$ (this is important for cloud computing, e.g., DevOps, library management, and modularity in general). The following proposition, \cref{thm:RedToUnit}, shows that our EPP is adequate.

\begin{proposition}\label{thm:RedToUnit}
Given a closed choreography, $M$, if $\Theta;\Sigma;\Gamma\vdash M:T$ and $R\notin \roles(M)\cup \roles(T)$ then $\llbracket M\rrbracket_R=\botm$.
\end{proposition}
\begin{proof}
Straightforward from \cref{thm:Induction,thm:ValUnit} and induction on $\Theta;\Sigma;\Gamma\vdash M:T$.
\end{proof}

We now show that there is a close correspondence between the executions of choreographies and of
their projections.
Intuitively, this correspondence states that a choreography can execute an action if, and only if, its projection can execute the same action, and both transition to new terms in the same relation.
However, this is not completely true: if a choreography $C$ reduces by \cref{rule:c-case}, then the result has fewer branches than the network obtained by performing the corresponding reduction in the projection of $C$.

In order to capture this, we revert to the notion of pruning~\cite{CHY12,CM20}, defined by $B\sqsupset B'$ iff $B\merge B'=B$.
Intuitively, if $B\sqsupset B'$, then $B$ offers the same and possibly more behaviours than $B'$.
This notion extends to networks by defining $\mathcal{N}\sqsupset\mathcal{N}'$ to mean that, for any role $\role{R}$, $\mathcal{N}(R)\sqsupset \mathcal{N}'(R)$.

\begin{example}
  \label{ex:prune}
  Consider the choreography
  \[
   C=\case{\Inl~\litAt{\unit}{R}}{x}{\select{R}{S}~\lab{left}~\litAt{0}{S}}{y}{\select{R}{S}~\lab{right}~\litAt{1}{S}} \, \text{.}
  \]
  Its projection for role $\role{S}$ is $\epp{C}{\role{S}}=\&_\role{R}\{\lab{left}:0,\lab{right}:1\}$.

  After entering the conditional in the choreography, $C$ reduces to
  $C'=\select{R}{S}~\lab{left}~\litAt{0}{S}$, whereas $S$ does not have a corresponding action and its
  behaviour remains $\&_\role{R}\{\lab{left}:0,\lab{right}:1\}$, which is not the projection of $C'$.
  However,
  $\&_\role{R}\{\lab{left}:0,\lab{right}:1\}\merge\&_\role{R}\{\lab{left}:0\}=\&_\role{R}\{\lab{left}:0,\lab{right}:1\}$,
  so $\epp{C'}{\role{S}}$ is a pruning of this behaviour.
  \eoe
\end{example}

In addition to pruning, we need to equate behaviours that differ only by applications to $\botm$ like $P$ and $(\lambda x:\botmt.P)~\botm$ introduced by the projection of applications.%
\begin{definition}\label{def:equiv}
We define $\equiv$ as the least equivalence relation on behaviours that is closed under context and $P\equiv (\lambda x:\botmt.P)~\botm$ for any behaviour $P$.
We write $\mathcal{N} \equiv \mathcal{N}'$ for the pointwise extension of $\equiv$ to networks (i.e., $\Pi_{R} R[P_R] \equiv \Pi_R R[P_R']$ iff $P_R\equiv P_R'$ for all $R$s) and $\mathcal{N}\sqsupseteq \mathcal{N}'$ if there is a network $\mathcal{N}''$ such that $\mathcal{N}\sqsupset \mathcal{N''}$ and $\mathcal{N}''\equiv \mathcal{N}'$.
\end{definition}%

We can finally show that the EPP of a choreography can do all that (completeness) and only what (soundness) the original choreography does.

\begin{theorem}[Completeness]\label{thm:ChorToNet}
  Given a closed choreography $M$, if $M\xrightarrow{\tau,\mathbf{R}}_D M'$ and $\Theta;\Sigma;\Gamma\vdash M:T$, then
  there exist networks $\mathcal N$ and $M''$ such that:
  $\epp M{}\rightarrow^+_{\epp D{}}\mathcal{N}$; $M'\rightarrow^*M''$; and
  $\mathcal{N}\sqsupseteq\epp{M''}{}$.
\end{theorem}
\begin{proof}
  By structural induction on the derivation of $M\rightarrow_D M'$.
\end{proof}

\begin{theorem}[Soundness]\label{thm:NetToChor}
  Given a closed choreography $M$, if $\Theta;\Sigma;\Gamma\vdash M:T$ and
  $\epp M{}\xrightarrow{\tau_{\mathbf{R}}}_{\DN} \mathcal{N}$ for some network $\mathcal N$, then
  there exist a choreography $M'$, a set mapping $D$, and a network $\mathcal N'$ such that: $M\rightarrow^*_D M'$; $\epp{D}{}=\DN$;
  $\mathcal{N}\rightarrow^*\mathcal{N}'$; and $\mathcal{N'}\sqsupseteq \epp{M'}{}$.
\end{theorem}
\begin{proof}
  By structural induction on $M$.
\end{proof}

From \cref{thm:ChorToNet,thm:NetToChor,thm:TypeReduc,thm:TypePres}, we get the following corollary, which states that a network derived from a
well-typed closed choreography can continue to reduce until all roles contain only local values.
\begin{corollary}[Deadlock-freedom]\label{thm:NetReduc}
  Given a closed choreography $M$ and a function environment $D$ containing all the functions of
  $M$, if $\Theta;\Sigma;\Gamma\vdash M:T$ and $\Theta;\Sigma;\Gamma\vdash D$, then: whenever
  $\epp M{}\rightarrow^*_{\epp D{}}\mathcal{N}$ for some network $\mathcal N$, either there exists $\mathbf{R}$ and $\mathcal{N'}$ such that
  $\mathcal{N}\xrightarrow{\tau_{\mathbf{R}}}_{\epp D{}}\mathcal{N}'$ or
  $\mathcal{N}=\prod_{R\in \roles(M)} R[L_R]$.
\end{corollary}

To complete this section, we give the EPPs of some of the choreographies found in \cref{sec:example}. Specifically, we show the projection of the creation of a secure channel and Diffie-Hellman onto the first role involved (both roles have similar behaviours in these choreographies), and the projection of our authentication protocol onto the client and server. Projecting our
choreographies $\funf{diffieHellman}(P,Q)$ and $\funf{makeSecureChannels}(P,Q)$
for role $\role P$ yields the following behaviours:

\begin{lstlisting}[mathescape=true]
$\epp{D(\fundef{diffieHellman}(P,Q))}1(Q) =
\fun{psk}{\typef{Int}}\;\fun{qsk}{\botmt}\;\fun{psg}{\typef{Int}}\;\fun{qsg}{\botmt}\;\fun{psp}{\typef{Int}}\;
\fun{qsp}{\botmt}$
  $\pair\;(\funf{modPow}_1\;psg\;(\recv{Q}\;\botm))\;psp)$
  $\hphantom{\pair\;}(\send{Q}\;(\funf{modPow}_1\;psg\;psk\;psp))$

$\epp{D(\fundef{makeSecureChannels}(P, Q))}1(Q) = \fun{key}{\pairt{\typef{Int}}{\botmt}}$
  $\Pair\;(\fun{val}{\typef{String}}\;((\snd\;key)\;(\send{Q}\;(\funf{encrypt}_1\;(\fst\;key)\;val)))$
  $\hphantom{\Pair}\;(\fun{val}{\botmt}\;(\funf{decrypt}_1\; (\fst\;key)\;
     (\recv{Q}\;(\snd\;key)))$
\end{lstlisting}

Note the way local functions such as $\funf{modPow}(P)$ in the choreography get projected to $\funf{modPow}_1$ at $P$, since they are treated as degenerate choreographies (they have only one role) and $P$ is the first and only role involved. Conversely, $\funf{modPow}(Q)$ gets projected as $\botm$ since it is located entirely at a different role.

In turn, $\funf{authenticate}(S, C, I)$ has the following projections for the
client and the server roles:

\begin{lstlisting}[mathescape=true]
$\epp{D(\fundef{authenticate}(S, C, I))}2(S, I) = \fun{credentials}{\typef{Credentials}}$
  $\fun{comcip}{\typef{String}\to\botmt}\;\fun{comipc}{\botmt\to \typef{String}}\;
  \fun{comips}{\botmt}$
    $((\fun{user}{\botmt}\;(\fun{salt}{\typef{String}}\;(\fun{hash}{\botmt}$
      $\&_{I}\{\lab{ok}:(\fun{token}{\botmt}\;\inl\;(\pair\;(comipc\; \botm)\;\botm))\;\botm,$
      $\hphantom{\&_{I}\{}\lab{ko}:\inr\;\botm\})$
    $\;\;\;\;\;(comcip\; (\funf{calcHash}_1 \; salt \; (\funf{password}_1 \; credentials))))$
    $\;\;\;(comipc\; \botm))$
    $\;(comcip\; (\funf{username}_1 \; credentials)))$
\end{lstlisting}

\begin{lstlisting}[mathescape=true]
$\epp{D(\fundef{authenticate}(S, C, I))}1(C,I) = \fun{credentials}{\botmt}$
  $\fun{comcip}{\botmt}\;\fun{comipc}{\botmt}\;\fun{comips}{\botmt\to \typef{String}}$
    $((\fun{user}{\botmt}\;(\fun{salt}{\botmt}\;(\fun{hash}{\botmt}$
       $\&_{\rolef{IP}}\{\lab{ok}:(\fun{token}{\botmt}\;\inl\;(\pair\;\botm\;
         (comips\; \botm)))\; \botm,$
       $\hphantom{\&_{\rolef{IP}}\{}\lab{ko}:\inr\;\botm\})$
       $\botm)\;\botm)\;\botm))$
\end{lstlisting}

Despite $S$ only being involved in receiving the final token, it still has to execute the applications of a number of $\botm$s, replicating the actions of $I$ and $C$.

\FloatBarrier

\section{Related Work}\label{sec:related}

We already discussed much of the previous and related work on choreographic languages and choreographic programming in \cref{sec:introduction}. In this section, we discuss relevant technical aspects in related work more in detail.

The language nearest to \chorlam is Choral~\cite{GMP20}, the first higher-order choreographic programming language. As we discussed in \cref{sec:introduction}, Choral comes with no formal explanation of its semantics, typing, projection, and guarantees. We have covered all of these in the present article, showing that it is possible to formulate a theory of higher-order choreographic programming that satisfies the expected properties of the paradigm (correctness of EPP and deadlock-freedom).
Our examples were chosen to illustrate that \chorlam is a satisfactory theory of the core principles of higher-order choreographies. Some examples, in particular \cref{ex:remfun,ex:choice} and the code in \cref{sec:seccom,sec:auth}, are reconstructions of the key examples given for Choral~\cite{GMP20}.
Since Choral is object-oriented and \chorlam is functional, there are two differences between code in Choral and terms in \chorlam terms: terms in \chorlam are pure, whereas code in Choral can have side-effects; and we use functions instead of object methods.
Both differences are orthogonal to the development of higher-order choreographies, so they fall out of the scope of this paper. However, they might represent interesting future work.

Our design choices for \chorlam were aimed at making our theoretical development as simple as possible, but without sacrificing details of practical importance (e.g., out-of-order execution, our types for tracking roles that participate in functions, and an adequate notion of EPP).
As a full-fledged implementation, Choral is not an ideal language to seek a clean foundational theory: it is an extension of Java, and therefore quite complicated for reasons orthogonal to choreographies.
Indeed, follow-up work on Choral already started using fragments of it to focus on particular aspects of the language (``mini'' versions)~\cite{GMPRSW21}. However, these do not come with any theoretical developments, as in here.

There exist theories of choreographies that support some form of higher-order composition, which are more restrictive than Choral and \chorlam~\cite{DH12,HG22}. In particular, they fall short of capturing distribution, both in terms of independent execution and data structures.
In~\cite{DH12}, the authors present a choreographic language for writing abstract specifications of system behaviour (as in multiparty session types~\cite{HYC16}) that supports higher-order composition. The language is much less expressive than \chorlam, e.g., it can express neither data nor computation.
More importantly, the design of the language hampers decentralisation: entering a choreography requires that the programmer picks a role as central coordinator, which then orchestrates the other roles with multicasts. This coordination effectively acts as a barrier, so processes cannot really perform their own local computations independently of each other when higher-order composition is involved.
After~\cite{DH12} and Choral~\cite{GMP20}, a theory of higher-order choreographic programming was proposed in~\cite{HG22}. While this theory supports computation at roles, it is even more centralised than~\cite{DH12}: every function application in a choreography requires that all processes generated by projection go through a global barrier that involves the entire system. The global barrier is modelled as a middleware in the semantics of the language, and involves even processes that do not contribute at all to the function or its arguments.

In contrast to~\cite{DH12} and~\cite{HG22}, \chorlam presents no ``hardcoded'' barriers: coordination among roles is left entirely to the programmer of the choreography, and EPP inserts no hidden synchronisations. Our EPP thus generates more concurrent and faithful implementations.
For example, consider the following choreography. \[\funf{Assem}(A)~(\com{C_1}{A}~\funf{Const}(C_1)~\com{O_1}{C_1}~\var{ord}@O_1)~(\com{C_2}{A}~\funf{Const}(C_2)~\com{O_2}{C_2}~\var{ord}@O_2)\]
This is a toy factory scenario, in which two roles $O_1$ and $O_2$ concurrently communicate orders to $C_1$ and $C_2$, which perform the construction of the ordered components before sending them to be assembled at $A$.
The choreography specifies 4 communications, and in \chorlam the implementation generated by EPP would faithfully perform exactly 4 synchronisations. Instead, in~\cite{HG22}, there would be 3 additional global synchronisations among all roles: two for the two invocations of $\funf{Const}$ (even if they are completely local to $C_1$ and $C_2$) and another when entering $\funf{Assem}$ (even if it is local to $A$).
These extra synchronisations are essential to the proof of correctness of EPP in~\cite{HG22}. Our theory dispenses from this restriction thanks to our semantics for application, which is novel for $\lambda$-calculus and supports out-of-order execution (\cref{fig:Sem,fig:rewrite}).

On a similar note, our \cref{thm:RedToUnit} guarantees that the code generated for a role by EPP contains no ``junk'', in the sense that no code is generated for a role if it is not involved in the source choreographic term.
This result is a validation of the modularity of our translation, and gives us a substitution principle of practical value: when a choreography is edited, only the code for the roles affected by the edit needs to be regenerated (by using EPP). This property can therefore be useful in DevOps practices, software maintenance, and versioning.
\Cref{thm:RedToUnit} depends on both our semantics for out-of-order execution (without it, we would need to insert code for global barriers) and the $\rho$-annotations for functional types used in higher-order parameters (without them, EPP could not statically determine which choreographic terms a role participates in, exactly).
A similar result does not hold for~\cite{HG22}, since these features are missing.

Previous theories of choreographies organised syntax in two layers: one for local computation and one for communication~\cite{CHY12,CM13,CMP21b,CMP21,CM20,DGGLM17,HG22,JV22}. The communication layer can typically refer to the computation layer, but not vice versa. \chorlam has a very different and novel design, whereby a unified language addresses both areas. This design is enabled by our formulation of the communication primitive as a function and our typing discipline, which allows for reasoning about terms that compose data and functions at different roles.
An important consequence of our unified approach is that \chorlam can express distributed data structures (e.g., pairs with elements located at different roles), which can be manipulated by independent local computations or in coordination by performing appropriate communications.
This feature is crucial for our examples in \cref{sec:example} (and several examples in the original presentation of Choral in~\cite{GMP20}).

Many of our examples and those in the presentation of Choral~\cite{GMP20} rely on function definitions that are parametric on roles, $f(\vec R)$. The same feature is supported in some previous works, but all of them are less expressive either because of the reasons explained above (e.g., barriers~\cite{DH12}) or they do not support higher-order choreographies~\cite{CM17}.
Role parametricity is a key feature in practical choreographic programming: not having it is limiting for modularity, as it is what allows us to define a choreography (e.g., for creating an encrypted channel) that can be modularly and uniformly reused for different pairs of participants of a larger choreography.

We believe that the principles explored within \chorlam can guide the future extension of choreographic languages. 
In particular, there are several implementations of choreographic languages that are equipped with ad-hoc, limited variations of choreographic procedures (functions in \chorlam) and could benefit from being extended with higher-order composition~\cite{CM13,HMBCY11,DGGLM17}.

Another related line of work is that on multitier programming and its progenitor calculus, Lambda 5~\cite{MCHP04}. Similarly to \chorlam, Lambda 5 and multitier languages have data types with locations~\cite{WWS20}. However, they are used very differently. In choreographic languages (thus \chorlam), programs have a ``global'' point of view and express how multiple roles interact with each other. By contrast, in multitier programming programs have the usual ``local'' point of view of a single role but they can nest (local) code that is supposed to be executed remotely.
The reader interested in a detailed comparison of choreographic and multitier programming can consult~\cite{GMPRSW21}, which presents algorithms for translating choreographies to multitier programs and vice versa. The correctness of these algorithms has never been proven, because they use an informally-specified fragment of Choral as a representative choreographic language. We conjecture that the introduction of \chorlam could be the basis for a future investigation of formal translations between choreographic programs (in terms of \chorlam) and multitier programs (in terms of Lambda 5).
In a similar direction, \cite{CY20} presented a simple first-order multitier language from which it is possible to infer abstract choreographies (computation is not included) that describe the communication flows that multitier programs enact. This language, like all existing multitier languages, does not fully support higher-order composition of multitier programs (programs cannot be values dynamically passed as arguments). Establishing translations between \chorlam and multitier languages might provide insight on how multitier languages can support higher-order composition.

Another notion of compositionality for choreographies is used in \cite{BDH19,MY13}, but it is orthogonal to higher-order choreographies. Specifically, in these works, choreographies can interact with each other via side-effects generated by local actions at processes.

\section{Conclusion and Future Work}\label{sec:conc}

We have presented \chorlam, a new theory of choreographic programming that supports higher-order, modular choreographies. \chorlam is equipped with a type system that guarantees progress (\cref{thm:TypeReduc,thm:NetReduc}).
We have then shown how to obtain correct implementations of choreographies in \chorlam in terms of a process language (\cref{thm:ChorToNet,thm:NetToChor}). Unlike any previous choreographic programming language, \chorlam is based on the $\lambda$-calculus. It therefore inherits the simple syntax of $\lambda$ terms and it is the first purely functional choreographic programming language. The semantics of \chorlam makes it the first theory of higher-order choreographies that is truly decentralised: processes can proceed independently unless the choreography specifies explicitly that they should interact.

\looseness=-1
We have demonstrated the usefulness of higher-order choreographies in \chorlam by modelling common protocols in \cref{sec:example}. The examples on single sign-on with encrypted channels and EAP, in particular, are parametrised on choreographies and cannot be expressed in previous theories, either because of lack of higher-order composition or because the semantics is not satisfactory due to global synchronisations---which the original protocol specifications do not expect.

\subparagraph{Future Work}
An obvious extension of \chorlam would be to add generic data types, which we did not include to keep the focus on choreographies. Since we use $\lambda$-calculus as foundation, we believe that this would be a straightforward import of known methods.

Other features that are interesting for \chorlam have been investigated in the context of first-order choreographic languages and represent future work.
These include: channel-based communication \cite{CM13}, dynamic creation of roles \cite{CM17}, internal threads \cite{CHY12}, group communication \cite{CMP18}, availability-awareness \cite{LNN16}, and runtime adaptation \cite{DGGLM17}.

A more sophisticated extension would be to reify roles, that is, extending the syntax such that values can be roles that can be acted upon. This could, for example, enable dynamic topologies: choreographies where a process receives at runtime a role that it needs to interact with at a later time.

Another interesting line of future work would be to extend existing formalisations of choreographic languages with the features explored in this work \cite{CMP21,CMP21b,HG22,PGSN22}.

\bibliography{biblio}

\appendix

\def\rulesvsep{5pt}

\section{Full definitions and proofs}
\label{sec:appendix}

\begin{definition}[Free Variables]\label{def:FreeVar}
Given a choreography $M$, the free variables of $M$, $\fv(M)$ are defined as:

\begin{tabular}{ll}
$\fv(N~N')=\fv(N)\cup\fv(N')$ & $\fv(\select{S}{R}~l~M)=\fv(M)$ \\
$\fv(x)=x$ &
$\fv(\lambda x:T.N)=\fv(N)\setminus \{x\}$ \\
  $\fv(\litAt{\unit}{R})=\emptyset$ & $\fv(\com{S}{R})=\emptyset$ \\
$\fv(f)=\emptyset$&
$\fv(\Pair~V~V')=\fv(V)\cup \fv(V')$ \\
\multicolumn{2}{l}{$\fv(\case{N}{x}{M}{y}{M'})=\fv(N)\cup (\fv(M)\setminus \{x\})\cup (\fv(M')\setminus \{y\})$}\\
$\fv(\fst)=\fv(\snd)=\emptyset$ & $\fv(\Inl~V)=\fv(\Inr~V)=\fv(V)$ \\
\end{tabular}
\end{definition}

\begin{figure}
\begin{spreadlines}{\rulesvsep}
\begin{eqnarray*}
& \infer[\rtag{\ref*{rule:t-abs}}]
	{\Theta;\Sigma;\Gamma\vdash\lambda x:T.M:T \rightarrow_\rho T'}
	{\roles(T \rightarrow_\rho T');\Sigma;\Gamma,x:T\vdash M:T' 
	&\roles(T \rightarrow_\rho T') \subseteq \Theta} 
\\ & 
\infer[\rtag{\ref*{rule:t-var}}]
 {\Theta;\Sigma;\Gamma\vdash x:T}
 {x:T\in \Gamma & \roles(T)\subseteq \Theta}
\qquad 
	\infer[\rtag{\ref*{rule:t-app}}]
	{\Theta;\Sigma;\Gamma\vdash N~M:T'}
	{\Theta;\Sigma;\Gamma\vdash N:T\rightarrow_\rho T' 
	&\Theta;\Sigma;\Gamma\vdash M:T} 
\\& 
\infer[\rtag{\rlabel{\rname{TCase}}{rule:t-case}}]
	{\Theta;\Sigma;\Gamma\vdash \case{N}{x}{M'}{x'}{M''}:T}
	{\Gamma\vdash N:T_1 + T_2 
	&\Theta;\Sigma;\Gamma,x:T_1\vdash M':T 
	&\Theta;\Sigma;\Gamma,x':T_2\vdash M'':T} 
\\ & 
\infer[\rtag{\ref*{rule:t-sel}}]
	{\Theta;\Sigma;\Gamma\vdash \select{S}{R}~l~M:T}
	{ \Theta;\Sigma;\Gamma\vdash M:T & S,R\in \Theta} 
\\ &
\infer[\rtag{\ref*{rule:t-fun}}]
		{\Theta;\Sigma;\Gamma\vdash f(\vec{\role R}):T[{\vec{R'}}:={\vec{R}}]}
		{f(\vec{\role{R'}}):T\in \Gamma & \roles(T)\subseteq \vec{\role R'}\subseteq\Theta & ||\vec{R}||=||\vec{R'}|| & \dist(\vec{R})}
\\& 
\infer[\rtag{\rlabel{\rname{TUnit}}{rule:t-unit}}]
	{\Theta;\Sigma;\Gamma\vdash \litAt{\unit}{R}:\unitt@R}
	{R\in\Theta} 
\qquad
\infer[\rtag{\ref*{rule:t-com}}]
	{\Theta;\Sigma;\Gamma\vdash \com{S}{R}:T \rightarrow_\emptyset T[S:=R]}
	{S,R\in \Theta 
	&\roles(T)=S}
\\& 
\infer[\rtag{\rlabel{\rname{TPair}}{rule:t-pair}}]
	{\Theta;\Sigma;\Gamma\vdash \Pair~V~V': (T\times T')}
	{\Theta;\Sigma;\Gamma\vdash V:T & \Theta;\Sigma;\Gamma\vdash V':T'} 
\\ & 
\infer[\rtag{\rlabel{\rname{TProj1}}{rule:t-proj1}}]
	{\Theta;\Sigma;\Gamma\vdash \fst:(T \times T') \rightarrow_\emptyset T}
	{\roles(T\times T')\subseteq \Theta}
\quad
\infer[\rtag{\rlabel{\rname{TProj2}}{rule:t-proj2}}]
	{\Theta;\Sigma;\Gamma\vdash \snd:(T \times T') \rightarrow_\emptyset T'}
	{\roles(T\times T')\subseteq \Theta} 
\\& 
\infer[\rtag{\rlabel{\rname{TInl}}{rule:t-inl}}]
	{\Theta;\Sigma;\Gamma\vdash \Inl~V:(T + T')}
	{\Theta;\Sigma;\Gamma\vdash V:T 
	&\roles(T+T')\subseteq \Theta} 
\quad
\infer[\rtag{\rlabel{\rname{TInR}}{rule:t-inr}}]
	{\Theta;\Sigma;\Gamma\vdash \Inr~V: (T + T')}
	{\Theta;\Sigma;\Gamma\vdash V:T' 
	&\roles(T+T')\subseteq \Theta} 
\\& 
\infer[\rtag{\ref*{rule:t-eq}}]
			{\Theta;\Sigma;\Gamma\vdash M:T[{\vec{R'}}:={\vec{R}}]}
			{\Theta;\Sigma;\Gamma\vdash M:t@\vec{R} 
			&t@\vec{R'}=_{\Sigma} T 
      &||\vec{R}||=||\vec{R'}|| 
      & \dist(\vec{R})}
\\&
\infer[\rtag{\ref*{rule:t-defs}}]
	{\Theta;\Sigma;\Gamma\vdash D}
		{\forall f(\vec{R})\in \mathsf{domain}(D): 
		& f(\vec{R}):T\in \Gamma 
		& \vec{R};\Sigma;\Gamma\vdash D(f(\vec{R})):T & \dist(\vec{R}) & \vec{R}\subseteq\Theta}
\end{eqnarray*}
\end{spreadlines}
\caption{Full set of typing rules for \chorlam.}\label{fig:Type}
\end{figure}

\begin{figure}
\begin{spreadlines}{\rulesvsep}
  \begin{eqnarray*}&	    
  \infer[{\rtag{\ref{rule:c-appabs}}}]
  {\lambda x:T.M~V\xrightarrow{\tau,\emptyset}_D M[x:=V]}
  {}
\qquad \infer[\rtag{\ref{rule:c-inabs}}]
     	{\lambda x:T.M\xrightarrow{\lambda,\mathbf{R}}_D \lambda x:T.M'}
     	{M\xrightarrow{\ell,\mathbf{R}}_D M' }
    &\\&
     \infer[\rtag{\ref{rule:c-app1}}]
     	{M~N\xrightarrow{\tau,\mathbf{R}}_D M'~N}
     	{M\xrightarrow{\ell,\mathbf{R}}_D M' & \ell=\lambda\Rightarrow\mathbf{R}\cap \roles(N)=\emptyset}
    &\\&
     \infer[\rtag{\ref{rule:c-app2}}]
     {V~N\xrightarrow{\tau,\mathbf{R}}_D V~N'}
     {N\xrightarrow{\tau,\mathbf{R}}_D N'} 
     \qquad  
     \infer[\rtag{\ref{rule:c-app3}}]
     {M~N\xrightarrow{\tau,\mathbf{R}}_D M~N'}
     {N\xrightarrow{\tau,\mathbf{R}}_D N' & \mathbf{R}\cap\roles(M)=\emptyset} 
    &\\&
     \infer[\rtag{\ref{rule:c-case}}]
     	{\case{N}{x}{M}{x'}{M'}\xrightarrow{\tau,\mathbf{R}}_D \case{N'}{x}{M}{x'}{M'}}
     	{N\xrightarrow{\tau,\mathbf{R}}_D N'} 
    &\\&
     \infer[\rtag{\ref{rule:c-incase}}]
     	{\case{N}{x}{M_1}{x'}{M_2}\xrightarrow{\ell,\mathbf{R}}_D \case{N}{x}{M_1'}{x'}{M_2'}}
     	{M_1\xrightarrow{\ell,\mathbf{R}}_D M_1' & M_2\xrightarrow{\ell,\mathbf{R}}_D M_2' & \mathbf{R}\cap \roles(N)=\emptyset} 
    &\\&
      \infer[\rtag{\ref{rule:c-casel}}]
     {\case{\Inl~V}{x}{M}{x'}{M'}\xrightarrow{\tau,\emptyset}_D M[x:= V]}{}
    &\\&
    \infer[\rtag{\rlabel{\rname{CaseR}}{rule:c-caser}}]
     {\case{\Inr~V}{x}{M}{x'}{M'}\xrightarrow{\tau,\emptyset}_D M'[x':= V]}{}
    &\\&
      \infer[\rtag{\ref{rule:c-proj1}}]
     {\fst~\Pair~V~V'\xrightarrow{\tau,\emptyset}_D V}{}
   \qquad
   \infer[\rtag{\rlabel{\rname{Proj2}}{rule:c-proj2}}]
   {\snd~\Pair~V~V'\xrightarrow{\tau,\emptyset}_D V'}{}
   &\\&
	    \infer[\rtag{\ref{rule:c-fun}}]{f(\vec{R}) \xrightarrow{\tau,\emptyset}_D M [{\vec{R'}}:={\vec{R}}]}{D(f(\vec{R'}))=M}
    &\\&
     \infer[{\rtag{\ref{rule:c-com}}}]{\com{S}{R}~V \xrightarrow{\tau,\{S,R\}}_D V[S:=R]}{\fv(V)=\emptyset}
    \qquad
    \infer[\rtag{\ref{rule:c-sel}} ]
     {\select{S}{R}~l~M \xrightarrow{\tau,\{S,R\}}_D M}{}
     &\\&
     \infer[\rtag{\ref{rule:c-insel}}]
     {\select{S}{R}~\ell~M \xrightarrow{\ell,\mathbf{R}}_D \select{S}{R}~\ell~M'}
     {M\xrightarrow{\ell,\mathbf{R}}_D M' & \mathbf{R}\cap \{S,R\}=\emptyset} 
     \qquad
     \infer[\rtag{\ref{rule:c-str}}]{M\xrightarrow{\tau,\mathbf{R}}_D M'}{M\rightsquigarrow^* N & N\xrightarrow{\tau,\mathbf{R}} N'}
  \end{eqnarray*}
  \end{spreadlines}
\caption{Semantics of \chorlam}
\label{fig:Sem}
\end{figure}

\begin{figure}[p]
\begin{spreadlines}{\rulesvsep}
  \begin{eqnarray*}&
 \infer[{\rtag{\ref{rule:n-send}}}]{\send{\role R}~L\xrightarrow{\send{\role R}~L}_{\DP} \botm}{\fv(L)=\emptyset}
    \qquad
    \infer[{\rtag{\ref{rule:n-recv}}}]
     {\recv{\role R}~\botm\xrightarrow{\recv{\role R}~L}_{\DP} L}{}
    \\&
     \infer[{\rtag{\ref{rule:n-com}}}]
     {S[B_1] \mid R[B_2]\xrightarrow{\tau_{\role S,\role R}}_{\DN} \role S[B'_1] \mid \role R[B'_2]}
     {B\xrightarrow{\send{\role S}~L}_{\DN(S)} B'_1 
     &B_2\xrightarrow{\recv{\role R}~L[\role S:=\role R]}_{\DN(R)} B'_2}
   \\&
   \infer[{\rtag{\ref{rule:n-cho}}}]
     {\oplus_{\role R}~l~B\xrightarrow{\oplus_{\role R}~l}_{\DP} B}
     {}
    \qquad
    \infer[{\rtag{\ref{rule:n-off}}}]
     {\&_{\role R}\{\ell_1:B_1,\dots,\ell_n:B_n\}\xrightarrow{\&_{\role R}\ell_i}_{\DP} B_i}{}
     \\&    
     \infer[{\rtag{\ref{rule:n-off2}}}]{\&_{\role R}\{\ell_1:B_1,\dots,\ell_n:B_n\}\xrightarrow{\mu}_{\DP} \&_{\role R}\{\ell_1:B_1',\dots,\ell_n:B_n'\}}{B_i\xrightarrow{\mu}_{\DP}B_i'\text{ for }1\leq i\leq n & \mu\in\{\tau,\lambda\}}
    \\&
     \infer[{\rtag{\ref{rule:n-cho2}}}]{\oplus_{\role R}~l~B\xrightarrow{\mu}_{\DP} \oplus_{\role R}~l~B'}{B\xrightarrow{\mu}_{\DP} B' & \mu\in\{\tau,\lambda\}}
     \qquad 
     \infer[{\rtag{\ref{rule:n-sel}}}]
     	{\role S[B_1] \mid R[B_2]\xrightarrow{\tau_{\role S,\role R}}_{\DN} \role S[B'_1] \mid \role R[B'_2]}
     	{B_1\xrightarrow{\oplus_{\role R}~\ell}_{\DN(S)} B'_1 
     	& B_2\xrightarrow{\&_{\role S}~\ell}_{\DN(R)} B'_2} 
    \\&
    \infer[{\rtag{\ref{rule:n-absapp}}}]
     {(\lambda x:T.B)~L\xrightarrow{\tau}_\DP B[x:=L]}
     {}
      \quad
     \infer[\rtag{\ref{rule:n-inabs}}]
     	{\lambda x:T.B\xrightarrow{\lambda}_D \lambda x:T.B'}
     	{B\xrightarrow{\mu}_\DP B' & \mu\in\{\tau,\lambda\} }
    \\&
     \infer[{\rtag{\ref{rule:n-app1}}}]
     	{B~B'\xrightarrow{\mu'}_\DP B''~B'}
     	{B\xrightarrow{\mu}_\DP B'' & \text{if }\mu=\lambda \text{ then } \mu'=\tau \text{ else } \mu'=\mu}
    \\&
     \infer[{\rtag{\ref{rule:n-app2}}}]
     	{L~B\xrightarrow{\mu}_\DP L~B'}
     	{B\xrightarrow{\mu}_\DP B'}  
     	\qquad
     \infer[{\rtag{\ref{rule:n-app3}}}]
     	{B~B'\xrightarrow{\tau}_\DP B~B''}
     	{B'\xrightarrow{\tau}_\DP B''}  \\&
     	\infer[{\rtag{\ref{rule:n-case}}}]
     	{\case{B}{x}{B'}{x'}{B''}\xrightarrow{\mu}_\DP \case{B'''}{x}{B'}{x'}{B''}}
     	{B\xrightarrow{\mu}_\DP B'''} \\&
     	\infer[{\rtag{\ref{rule:n-case2}}}]
     	{\case{B}{x}{B_1}{x'}{B_2}\xrightarrow{\mu}_\DP \case{B}{x}{B_1'}{x'}{B_2'}}
     	{B_1\xrightarrow{\mu}_\DP B_1' & B_2\xrightarrow{\mu}_\DP B_2' & \mu\in\{\lambda,\tau\}}       \\&
       \infer[{\rtag{\rlabel{\rname{NCaseL}}{rule:n-casel}}}]
     	{\case{\Inl~L}{x}{B}{x'}{B'}\xrightarrow{\tau}_\DP B[x:= L]}{}
    \\&
    \infer[{\rtag{\rlabel{\rname{NCaseR}}{rule:n-caser}}} ]
     {\case{\Inr~L}{x}{B}{x'}{B'}\xrightarrow{\tau}_\DP B'[x':=L]}{}
    \\&
    \infer[{\rtag{\rlabel{\rname{NProj1}}{rule:n-proj1}}}]
    	{\fst~\Pair~L~L'\xrightarrow{\tau}_\DP L}{}
    \qquad
    \infer[{\rtag{\rlabel{\rname{NProj2}}{rule:n-proj2}}} ]
    	{\snd~\Pair~L~L'\xrightarrow{\tau}_\DP L'}
      {}
    \\&
     \infer[{\rtag{\ref{rule:n-pro}}}]
     	{\role R[B]\xrightarrow{\tau_R}_{\DN} R[B']}
     	{B\xrightarrow{\tau}_{\DN(R)} B'}
    \qquad
     	\infer[{\rtag{\ref{rule:n-par}}}]
     		{\mathcal{N}\mid\mathcal{N}'\xrightarrow{\tau_{\mathbf{R}}}_{\DN} \mathcal{N}''\mid\mathcal{N}'}
     		{\mathcal{N}\xrightarrow{\tau_{\mathbf{R}}}_{\DN} \mathcal{N}''} \\&	
     \infer[\rtag{\rlabel{\rname{Nfun}}{rule:n-fun}}]{f(\vec{R}) \xrightarrow{\tau}_\DP B [{\vec{R'}}:={\vec{R}}]}{D(f(\vec{R'}))=B}
     \qquad
     \infer[\rtag{\ref{rule:n-str}}]{B\xrightarrow{\mu}_\DP B'}{B\rightsquigarrow^* B'' & B''\xrightarrow{\mu} B'}  
  \end{eqnarray*}
  \end{spreadlines}
  \caption{Semantics of networks.}
  \label{fig:NetSem}
\end{figure}

\begin{figure}[t]
	\begin{spreadlines}{\rulesvsep}
  \begin{eqnarray*}&
     \infer[{\rtag{\ref{rule:nt-cho}}}]
     	{\Sigma;\Gamma\vdash \oplus_R~\ell~B:T}
     	{ \Sigma;\Gamma\vdash B:T}
    \qquad
     \infer[{\rtag{\ref{rule:nt-off}}}]
     	{\Sigma;\Gamma\vdash \&_R\{\ell_1:B_1,\dots\ell_n:B_n\}:T}
     	{\Sigma;\Gamma\vdash B_i:T \text{ for }1\leq i\leq n} 
    \\&
    \infer[{\rtag{\ref{rule:nt-send}}}]
     {\Sigma;\Gamma\vdash \send{R}:T \rightarrow \botmt}{}
    \qquad
    \infer[{\rtag{\ref{rule:nt-recv}}}]
     {\Sigma;\Gamma\vdash \recv{R}:\botmt \rightarrow T}
     {}
		\\& 
			\infer[{\rtag{\rlabel{\rname{NTAbs}}{rule:nt-Abs}}}]
				{\Sigma;\Gamma\vdash\lambda x:T.B:T \rightarrow T'}
				{\Sigma;\Gamma,x:T\vdash B:T'} 
		\qquad
			\infer[{\rtag{\rlabel{\rname{NTVar}}{rule:nt-var}}}]
				{\Sigma;\Gamma\vdash x:T}
				{x:T\in \Gamma}
		\\& 
			\infer[{\rtag{\rlabel{\rname{NTApp}}{rule:nt-app}}}]
				{\Sigma;\Gamma\vdash B~B':T'}
				{\Sigma;\Gamma\vdash B:T\rightarrow T' 
				&\Sigma;\Gamma\vdash B:T} 
				\qquad
     \infer[{\rtag{\ref{rule:nt-app2}}}]
     	{\Sigma;\Gamma\vdash B~B':\botmt}
     	{ \Sigma;\Gamma\vdash B:\botmt & \Sigma;\Gamma\vdash B':\botmt}
		\\& 
			\infer[{\rtag{\rlabel{\rname{NTCase}}{rule:nt-case}}}]
				{\Sigma;\Gamma\vdash \case{B}{x}{B'}{x'}{B''}:T}
				{\Sigma;\Gamma\vdash B:T_1 + T_2 & \Sigma;\Gamma,x:T_1\vdash B':T & \Sigma;\Gamma,x':T_2\vdash B'':T} 
		\\& 
			\infer[{\rtag{\rlabel{\rname{NTDef}}{rule:nt-fun}}}]
				{\Sigma;\Gamma\vdash f:T}
				{f:T\in \Gamma } 
		\qquad
      \infer[{\rtag{\rlabel{\rname{NTUnit}}{rule:nt-unit}}}]
			{\Sigma;\Gamma\vdash \unit:\unitt}{}
		\qquad
    \infer[{\rtag{\ref{rule:nt-botm}}} ]
			{\Sigma;\Gamma\vdash \botm:\botmt}
      {}
		\\& 
    \infer[{\rtag{\rlabel{\rname{NTPair}}{rule:nt-pair}}}]
			{\Sigma;\Gamma\vdash \Pair:T \rightarrow_\emptyset T' \rightarrow_\emptyset (T\times T')}{}
		\\& 
    \infer[{\rtag{\rlabel{\rname{NTProj1}}{rule:nt-proj1}}}]
			{\Sigma;\Gamma\vdash \fst:(T \times T') \rightarrow_\emptyset T}
      {}
		\qquad
    \infer[{\rtag{\rlabel{\rname{NTProj2}}{rule:nt-proj2}}}]
			{\Sigma;\Gamma\vdash \snd:(T \times T') \rightarrow_\emptyset T'}{}
		\\& 
			\infer[{\rtag{\rlabel{\rname{NTEq}}{rule:nt-eq}}}]
				{\Sigma;\Gamma\vdash B:T}
				{\Sigma;\Gamma\vdash B:T' 
				&\{T=T',T'=T\}\cap\Sigma\neq\emptyset} 
		\\&
			\infer[{\rtag{\rlabel{\rname{NTDefs}}{rule:nt-defs}}}]
				{\Sigma;\Gamma\vdash \DP}
				{\forall f\in \mathsf{domain}(\DP) 
				&f:T\in \Gamma
				&\Sigma;\Gamma\vdash \DP(f):T}
  \end{eqnarray*}
  \end{spreadlines}
	\label{fig:ProcType}
	\caption{Typing rules for simple processes.}
\end{figure}

\begin{definition}[Merging]\label{def:merge}
Given two behaviours $B$ and $B'$, $B\merge B'$ is defined as follows.
\begin{eqnarray*}
& B_1~B_2\merge B'_1~B'_2= (B_1\merge B'_1)~(B_2\merge B'_2) \\
& \case{B_1}{x}{B_2}{y}{B_3}\merge \case{B'_1}{x}{B'_2}{y}{B'_3}= \\
& \quad \case{(B_1\merge B'_1)}{x}{(B_2\merge B'_2)}{y}{(B_3\merge B'_3)} \\
& \oplus_R~\ell~B \merge \oplus_R~\ell~B'=\oplus_R~\ell~(B\merge B') \\
& \&\{\ell_i:B_i\}_{i\in I} \merge \&\{\ell_j:B'_j\}_{j\in J}= \&\left(\{\ell_k:B_k\merge B'_k\}_{k\in I\cap J} \cup \{\ell_i:B_i\}_{i\in I\setminus J}\cup \{\ell_j:B'_j\}_{j\in J\setminus I}\right) \\
& x\merge x=x \quad\quad
\lambda x:T.B\merge \lambda x:T.B'=\lambda x:T.(B\merge B') \\
& \fst \merge \fst=\fst \quad\quad \snd\merge \snd=\snd \\
& \Inl~L \merge \Inl~L' =\Inl~(L\merge L') \quad\quad  \Inr~L \merge \Inr~L' =\Inr~(L\merge L') \\
& \Pair~L_1~L_2 \merge \Pair~L'_1~L'_2 =\Pair~(L_1\merge L_1')~(L_2\merge L'_2)\quad\quad
f\merge f=f \\
& \recv{R}\merge \recv{R}= \recv{R}\quad\quad
\send{R}\merge \send{R}=\send{send}{R}
\end{eqnarray*}
\end{definition}

\begin{figure}\begin{footnotesize}\begin{eqnarray*}
\lefteqn{\mbox{Choreographies:}}\\
    && \epp{M~N}R =
    \begin{dcases*}\epp MR~\epp NR & if $R\in \roles(\type(M))$ or $R\in\roles(M)\cap\roles(N)$ \\
    \botm & if $\epp MR=\epp NR=\botm$ \\ 
    \epp MR & if $\epp NR=\botm$ \\
    \epp NR & otherwise \\
\end{dcases*} \\[5pt]
    && \epp{\lambda x:T.M}R =
    \begin{dcases*} \lambda x.\epp MR & if $R\in \roles(\type(x:T.M))$ \\
      \botm & otherwise \end{dcases*} \\[5pt]
    && \epp{\case{M}{x}{N}{x'}{N'}}R = \\[5pt]
    && \qquad
    \begin{dcases*} \case{\epp MR}{x}{\epp NR}{x'}{\epp{N'}R} & if $R\in \roles(\type(M))$ \\
    \botm & if $\epp MR= \epp NR= \epp {N'}R=\botm$ \\
    \epp MR & if $\epp NR=\epp {N'}R=\botm$ \\
    \epp NR \merge \epp {N'}R & if $\epp MR=\botm$ \\
      (\lambda x'':\botmt.\epp NR \merge \epp{N'}R)~\epp MR & for some $x''\notin {\fv(N) \cup \fv(N')}$ \\ & otherwise \end{dcases*} \\[5pt]
    && \epp{\select{S}{S'}~\ell~M}R =
    \begin{dcases*}\oplus_{S'}~\ell~\epp MR & if $R=S\neq S'$ \\
      \&_S\{\ell:\epp MR\} & if $R=S'\neq S$ \\
      \epp MR & otherwise \end{dcases*} \\[5pt]
    && \epp{\com{S}{S'}}R =
    \begin{dcases*} \lambda x.x & if $R=S=S'$ \\
      \send{S'} & if $R=S\neq S'$\\
      \recv{S} & if $R=S'\neq S$\\
      \botm & otherwise \end{dcases*} \\
    && \epp{\litAt{\unit}{S}}R= \begin{dcases*}
    \unit & if $S=R$ \\
    \botm & otherwise
    \end{dcases*}
    \qquad\qquad
    \epp xR =
    \begin{dcases*} x & if $R\in \roles(\type(x))$ \\ \botm & otherwise \end{dcases*}
    \\
    &&
    \epp {f(\vec{R})}R=
    \begin{dcases*}f_i( R_1,\dots,R_{i-1},R_{i+1},\dots, R_n) & if $\vec{R}= R_1,\dots,R_{i-1},R,R_{i+1},\dots, R_n$ \\
    \botm & otherwise
    \end{dcases*} \\
&& \llbracket \Pair~V~V' \rrbracket_R= \begin{dcases*} \Pair~\llbracket V \rrbracket_R~\llbracket  V' \rrbracket_R & if $R\in\roles(\type(V)\times \type(V'))$ \\
\botm & otherwise \end{dcases*} \\[5pt]
&& \llbracket \fst \rrbracket_R= \begin{dcases*} \fst & if $R\in\roles(\type(\fst))$ \\
\botm & otherwise \end{dcases*} \qquad
 \llbracket \snd \rrbracket_R= \begin{dcases*} \snd & if $R\in\roles(\type(\snd))$ \\
\botm & otherwise \end{dcases*} \\[5pt]
&& %
 \llbracket \Inl~V \rrbracket_R= \begin{dcases*} \Inl~\llbracket V\rrbracket_R & if $R\in\roles(\type(\Inl~V))$ \\
\botm & otherwise \end{dcases*} %
\quad \llbracket \Inr~V \rrbracket_R= \begin{dcases*} \Inr~\llbracket V\rrbracket_R & if $r\in\roles(\type(\Inr~V))$ \\
\botm & otherwise \\
\end{dcases*}  \\[5pt]
\lefteqn{\mbox{Types:}}\\
&& \llbracket T\rightarrow_\rho T' \rrbracket_R= \begin{dcases*}\llbracket T \rrbracket_R \rightarrow \llbracket T' \rrbracket_R & if $R\in \rho\cup\roles(T)\cup \roles(T')$ \\
\botmt & otherwise
\end{dcases*} \qquad
\llbracket \unitt@S \rrbracket_R= \begin{dcases*}
    \unitt & if $S=R$ \\
    \botmt & otherwise
    \end{dcases*}\\[5pt]
&& \llbracket T\times T' \rrbracket_R= \begin{dcases*} \llbracket T \rrbracket_R \times \llbracket  T' \rrbracket_R & if $R\in\roles(T\times T')$ \\
\botmt & otherwise \end{dcases*}   \quad \llbracket T+ T' \rrbracket_R= \begin{dcases*} \llbracket T \rrbracket_R + \llbracket  T' \rrbracket_R & if $R\in\roles(T+ T')$ \\
\botmt & otherwise \end{dcases*} \\[5pt]
&&
\epp{t@\vec{R}}R =
    \begin{dcases*} t_i & if $\vec{R}= R_1,\dots,R_{i-1},R,R_{i+1},\dots, R_n$ \\ \botmt & otherwise \end{dcases*} \\[5pt]
    \lefteqn{\mbox{Definitions:}} \\
    && \epp D{} = \{f_i( R_1,\dots, R_{i-1},R_{i+1},\dots, R_n) \mapsto \epp{D(f( R_1,\dots, R_n))}{R_i} \mid f( R_1,\dots, R_n)\in\mathsf{domain}(D) \}\}
\end{eqnarray*}\end{footnotesize}\caption{Projecting \chorlam onto a role}\label{fig:Proj}
\end{figure}

\begin{definition}[Context]\label{def:context}
We define a context $C[]$ in \chorlam as follows:
\[\begin{array}{rl}
C[]\Coloneqq& [] \mid M~C[] \mid C[]~M \mid \select{\role R}{\role R}~l~C[] \mid 
\case {C[]} xM{x}{M} \\ & \mid \case M x{C[]}{x}{M} \mid \case M x{M}{x'}{C[]} \mid
\lambda{x:T}.C[]
\end{array}\]
\end{definition}

\FloatBarrier
\section{Proof of \cref{thm:TypeReduc}}\label{app:TypeReduc}

\begin{proof}[Proof of \cref{thm:TypeReduc}]
We prove this by induction on the typing derivation of $\Theta;\Sigma;\Gamma\vdash M:T$. Most cases either $M$ is a value, or the result follows from simple induction, we go through the rest.
\begin{itemize}
\item Assume we use \cref{rule:t-app}, so $M=N_1~N_2$, $\Theta;\Sigma;\Gamma\vdash N_1:T'\rightarrow_\rho T$, and $\Theta;\Sigma;\Gamma\vdash N_2:T'$. If $N_1$ or $N_2$ is not a value then the result follows from induction and using \cref{rule:c-app1} or \cref{rule:c-app2}. Otherwise, we have four cases:
\begin{itemize}
\item Assume $\Theta;\Sigma;\Gamma\vdash N_1:T'\rightarrow_\rho T$ uses \cref{rule:t-abs}. Then the result follows using \cref{rule:c-appabs}.
\item Assume $\Theta;\Sigma;\Gamma\vdash N_1:T'\rightarrow_\rho T$ uses \cref{rule:t-com}. Then, since $M$ is closed, the result follows by \cref{rule:c-com}.
\item Assume $\Theta;\Sigma;\Gamma\vdash N_1:T'\rightarrow_\rho T$ uses \cref{rule:t-proj1}. Then, since $M$ is closed and $N_2$ is a value, $N_2=\Pair~V~V'$, and consequently the result follows using \cref{rule:c-proj1}.
\item Assume $\Theta;\Sigma;\Gamma\vdash N_1:T'\rightarrow_\rho T$ uses \cref{rule:t-proj2}. Then, since $M$ is closed and $N_2$ is a value, $N_2=\Pair~V~V'$, and consequently the result follows using \cref{rule:c-proj2}.
\end{itemize}
\item Assume we use \cref{rule:t-case}, so $M=\case{N_1}{x}{N_2}{x'}{N_3}$, $\Theta;\Sigma;\Gamma\vdash N_1:T_1+T_2$, $\Theta;\Sigma;\Gamma,x:T_1\vdash N_2:T$, and $\Theta;\Sigma;\Gamma,x':T_2\vdash N_3:T$. Then if $N_1$ is not a value the result follows from induction and using \cref{rule:c-case}. If $N_1$ is a value then, since $M$ is closed, either $N_1=\Inl~V$ or $N_1=\Inr~V$, and the result follows by \cref{rule:c-casel} or \cref{rule:c-caser} respectively.
\item Assume we use \cref{rule:t-sel} so $M=\select{S}{R}~l~N:T$, $\Theta;\Sigma;\Gamma\vdash N:T$, and $S,R\in\Theta$. Then the result follows from using \cref{rule:c-sel}.
\item Assume we use \cref{rule:t-fun} and $M=f(\vec{R})$, $f(\vec{R'}:T'\in\Gamma$, $\vec{\role R}\subseteq \Theta$, $||\vec{R}||=||\vec{R'}||$, $\dist(\vec{R})$, and $T=T'[\vec{R'}:=\vec{R}]$. Then the result follows from $D$ containing $f$, $\Theta;\Sigma;\Gamma\vdash D$ and \cref{rule:c-fun}. 
\end{itemize}
\end{proof}
\section{Proof of \cref{thm:TypePres}}\label{app:TypePres}
\begin{lemma}\label{thm:ThetaExp}
Given a choreography, $M$, if $\Theta;\Sigma;\Gamma\vdash M$ then $\Theta\cup\Theta';\Sigma;\Gamma\vdash M$
\end{lemma}
\begin{proof}
Follows from the typing rules only ever discussing subsets of $\Theta$.
\end{proof}
\begin{lemma}[Type preservation under rewriting]\label{thm:TypePresRew}
Let $M$ be a choreography.
  If there exists a typing context $\Theta;\Sigma;\Gamma$ such that $\Theta;\Sigma;\Gamma\vdash M:T$, then $\Theta;\Sigma;\Gamma\vdash M':T$ for any $M'$ such that $M\rightsquigarrow M'$.
\end{lemma}
\begin{proof}
We prove this by case analysis of the the rewriting rules:
\begin{description}
\item[\cref{rule:r-absr}] Then $M=((\lambda x:T_1.N_1)~N_2)~N_3$, and from the typing rules we get that there exist $T_2$ and $\rho$ such that $\Theta;\Sigma;\Gamma,x:T_1\vdash N_1:T_2\rightarrow_\rho T$, $\Theta;\Sigma;\Gamma\vdash N_2:T_1$, and $\Theta;\Sigma;\Gamma\vdash N_3:T_2$, and $M'=(\lambda x:T_1.(N_1~N_3))~N_2$. The result follows from using \cref{rule:t-app,rule:t-abs} and $x\notin \fv(N_3)$.
\item[\cref{rule:r-absl}] Then $M=N_1~((\lambda x:T_1.N_3)~N_2)$, $\Theta;\Sigma;\Gamma,x:T_1\vdash N_1:T_2\rightarrow_\rho T$, $\Theta;\Sigma;\Gamma\vdash N_2:T_1$, and $\Theta;\Sigma;\Gamma\vdash N_3:T_2$, and $M'=(\lambda x:T_1.(N_1~N_3))~N_2$. The result follows from using \cref{rule:t-app,rule:t-abs} and $x\notin \fv(N_1)$.
\item[\cref{rule:r-caser}] Then $M=\case{N_1}{x}{N_2}{x'}{N_3})~N_4$, and from the typing rules we get that there exist $T_1$, $T_2$, $T_3$, and $\rho$ such that $\Theta;\Sigma;\Gamma\vdash N_1:T_1 + T_2$, $\Theta;\Sigma;\Gamma,x:T_1\vdash N_2:T_3\rightarrow_\rho T$, $\Theta;\Sigma;\Gamma,x':T_2\vdash N_3:T_3\rightarrow_\rho T$, and $\Theta;\Sigma;\Gamma\vdash N_4:T_3$, and $M'=\case{N_1}{x}{N_2~N_4}{x'}{N_3~N_4})$. The result follows from using \cref{rule:t-case,rule:t-abs} and $x,x'\notin \fv(N_4)$.
\item[\cref{rule:r-casel}] This case is similar to the previous.
\item[\cref{rule:r-selr}] Then $M=N_1~(\select{S}{R}~l~N_2)$, and from the typing rules we get that there exist $T'$, and $\rho$ such that $\Theta;\Sigma;\Gamma\vdash N_1:T'\rightarrow_\rho T$ and $\Theta;\Sigma;\Gamma\vdash N_2:T'$, and $M'=\select{S}{R}~l~(N_1~N_2)$. The result follows from using \cref{rule:t-sel,rule:t-abs}.
\item[\cref{rule:r-sell}]This case is similar to the previous.
\end{description}
\end{proof}

\begin{proof}[Proof of \cref{thm:TypePres}]
We prove this by induction on the derivation of $M\xrightarrow{\tau,\mathbf{R}}_D M'$. The cases for \cref{rule:c-appabs,rule:c-app1,rule:c-app2} are standard for simply-typed $\lambda$-calculus. And the cases for \cref{rule:c-inabs,rule:c-app3,rule:c-case,rule:c-incase,rule:c-insel} follow from simple induction. We go through the rest.
\begin{itemize}
\item Assume we use \cref{rule:c-casel}. Then we know that $M=\case{\Inl~V}{x}{N_1}{x'}{N_2}$, and from the typing rules we get that there exists $T'$ such that $\Theta;\Sigma;\Gamma\vdash V:T'$ and $\Theta;\Sigma;\Gamma,x:T'\vdash N_1:T$. Therefore, $\Theta;\Sigma;\Gamma\vdash N_1[x:=V]:T$.
\item Assume we use \cref{rule:c-caser}. This is similar to the previous case.
\item Assume we use \cref{rule:c-proj1}. Then we know $M=\fst~\Pair~V~V'$, and from the typing rules we get that $\Theta;\Sigma;\Gamma\vdash V:T$.
\item Assume we use \cref{rule:c-proj2}. This is similar to the previous case.
\item Assume we use \cref{rule:c-fun}. From the typing of $M$ we get that there exists $f(\vec{R'}):T\in\Gamma$ such that $||\vec{R}||=||\vec{R'}||$, $\vec{R}\subseteq \Theta$, and $\dist(\vec{R'})$. From the typing of $D$ we get that $\dist(\vec{R'})$ and $\vec{R'};\Sigma;\Gamma\vdash D(f(\vec{R'}))$. Therefore, by \cref{thm:ThetaExp}, we get $\Theta;\Sigma;\Gamma\vdash D(f(\vec{R'}))$.
\item Assume we use \cref{rule:c-com}. Then we know that $M=\com{S}{R}~V$, $\fv(V)=\emptyset$, and there exists $T'$ such that $\Theta;\Sigma;\Gamma\vdash V:T'$, $\roles(T')=\{S\}$, and $T=T'[S:=R]$. We see from our typing rules that the only time we use roles not mentioned in the choreography in typing is when handling free variables. Therefore we get that $\Theta;\Sigma;\Gamma\vdash V[S:=R]:T$.
\item Assume we use \cref{rule:c-sel}. Then we know that $M=\select{S}{R}~l~N$ and $\Theta;\Sigma;\Gamma\vdash N:T$. The result follows.
\item Assume we use \cref{rule:c-str}. Then the result follows from \cref{thm:TypePresRew} and induction.
\end{itemize}
\end{proof}

\subsection{Proof of \Cref{thm:ChorToNet}}\label{pro:ChorToNet}

\begin{lemma}\label{thm:Induction}
Given a choreography $M$, if $\Theta;\Sigma;\Gamma\vdash M:T$ then for any role $\role{R}$ in $\roles(M)$, $\llbracket M\rrbracket_R=L$ if $M=V$.
\end{lemma}
\begin{proof}
Straightforward from the projection rules.
\end{proof}

\begin{lemma}\label{thm:TypeUnit}
Given a type $T$, for any role $R\notin \roles(T)$, $\llbracket T\rrbracket_R=\botmt$.
\end{lemma}
\begin{proof}
Straightforward from induction on $T$.
\end{proof}

\begin{lemma}\label{thm:ValUnit}
Given a value $V$, for any role $R\notin \roles(\type(V))$, we have $\llbracket V\rrbracket_R=\botm$.
\end{lemma}
\begin{proof}
Follows from \cref{thm:Induction,thm:TypeUnit} and the projection rules.
\end{proof}

\begin{lemma}\label{thm:squignorestrict}
If $M\rightsquigarrow M'$ and $M\xrightarrow{\tau,\mathbf{R}}_D M''$ then $M'\xrightarrow{\tau,\mathbf{R}}_D M'''$ such that $M''\rightsquigarrow^* M'''$
\end{lemma}
\begin{proof}
Follows from case analysis on $M\rightsquigarrow M'$.
\end{proof}

\begin{lemma}\label{thm:squigepp}
If $M\rightsquigarrow M'$ then for any role $R$, $\epp{M}{R}\rightsquigarrow\cup\xrightarrow{\tau}^* B$ such that $B\equiv \epp{M'}{R}$
\end{lemma}
\begin{proof}
Follows from case analysis on $M\rightsquigarrow M'$.
\end{proof}

\begin{proof}[Proof of \Cref{thm:ChorToNet}]
We prove this by structural induction on $M\xrightarrow{\tau,\mathbf{R}}_D M'$.
\begin{itemize}
\item Assume $M=\lambda x:T.N~V$ and $M'=N[x:=V]$. Then for any role $R\in \roles(\type(\lambda x:T.N))$, we have $\llbracket M\rrbracket_R=(\lambda x:\epp{T}{R}. \llbracket N\rrbracket_R)~\llbracket V\rrbracket_R$ and $\llbracket M'\rrbracket_R=\llbracket N\rrbracket_R[x:=\llbracket V\rrbracket_R]$, and for any $R'\notin \roles(\type(\lambda x:T.N))$, we have $R'\notin \roles(\type(V))$ and therefore $\llbracket M\rrbracket_{R'}=\epp {M'}{R'}=\botm$. We therefore get $R[\llbracket M\rrbracket_R]\xrightarrow{\tau}_{\llbracket D\rrbracket} \llbracket M'\rrbracket_R$ for all $R\in\roles(\type(\lambda x:T.N)$ and define $\mathcal{N}=\prod\limits_{R\in\roles(\type(\lambda x:T.N))}R[\llbracket M'\rrbracket_R]\mid\prod\limits_{R'\notin\roles(\type((\lambda x:T.N))} R'[\botm]$ and the result follows.

\item Assume $M=N~M''$, $M'=N'~M''$, and $N\xrightarrow{\tau,\mathbf{R}}_D N'$. Then for any role $R\in \roles(\type(N))$, $\llbracket M\rrbracket_R= \llbracket N\rrbracket_R~\llbracket M''\rrbracket_R$ and $\llbracket M'\rrbracket_R=\llbracket N'\rrbracket_R~\llbracket M''\rrbracket_R$. 
For any role $R'$ such that $\epp N{R'}=\epp{M''}{R'}=\botm$, by induction we have $\epp{N'}{R'}=\botm$, and therefore $\epp M{R'}=\epp{M'}{R'}=\botm$. 
For any other role $R''$ such that $\epp{N}{R''}=\botm$, by induction we get $\epp{N'}{R''}=\botm$ and therefore $\epp M{R''}=\epp{M'}{R''}=\epp{M''}{R''}$. 
For any other role $R'''$ such that $\epp{M''}{R'''}=\botm$, we get $\epp{M}{R'''}=\epp{N}{R'''}$ and $\epp{M'}{R'''}=\epp{N'}{R'''}$. 
And by induction $\llbracket N\rrbracket\rightarrow^*_{\llbracket D\rrbracket} \mathcal{N}_N$ and $N'\rightarrow^*_{\llbracket D\rrbracket} N''$ for $\mathcal{N}_N\sqsupseteq\epp{N''}{}$. 
For any role $\role{R}$ we therefore get $\llbracket N\rrbracket_R\xrightarrow{\mu_0}_{\llbracket D\rrbracket}\xrightarrow{\mu_1}_{\llbracket D\rrbracket}\dots B_R$ for $B_R\sqsupseteq \epp{N''}{R}$ for some sequences of transitions $\xrightarrow{\mu_0}_{\llbracket D\rrbracket}\xrightarrow{\mu_1}_{\llbracket D\rrbracket}\dots$, and from the network semantics we get 
\[\llbracket M\rrbracket \rightarrow^* \begin{array}{l}
\prod\limits_{R\in \roles(\type(N))\cup(\roles(N)\cap\roles(M''))}R[B_R~\llbracket M''\rrbracket_R] \mid \prod\limits_{\epp N{R'}=\epp{M''}{R'}=\botm} R'[\botm] \\ \mid \prod\limits_{\epp{M}{R''}=\epp{M''}{R''}} R''[\epp{M''}{R''}] \mid \prod\limits_{\epp{M}{R'''}=\epp{N}{R'''}} R''[B_{R''}] \\ %
\end{array}=\mathcal{N}\] and 
$  M' \rightarrow^* N''~M$.
And since $\llbracket N\rrbracket\rightarrow^*_{\llbracket D\rrbracket} \mathcal{N}'$ and $\llbracket N'\rrbracket\rightarrow^*_{\llbracket D\rrbracket} \mathcal{N}'_N$, we know these sequences of transitions can synchronise when necessary, and if $\epp N{R''''}\neq  \epp{N'}{R''''}=\botm$ then we can do the extra application to get rid of this unit.

\item Assume $M=V~N$, $M'=V~N'$, and $N\xrightarrow{\tau,\mathbf{R}}_D N'$. Then for any role $R\in \roles(\type(V))$, $\llbracket M\rrbracket_R= \llbracket V\rrbracket_R~\llbracket N\rrbracket_R$ and $\llbracket M'\rrbracket_R=\llbracket V\rrbracket_R~\llbracket N'\rrbracket_R$. Since $V$ is a value, for any role $R'\notin\roles(\type(V))$, we have $\epp V{R'}=\botm$ and so for any role $R'$ such that $\epp V{R'}=\epp{N}{R'}=\botm$, by induction we get $\epp{N'}{R'}=\botm$ and therefore $\epp{M}{R'}=\epp{M'}{R'}=\botm$. For any other role $R''$ such that $\epp{V}{R''}=\botm$, we have $\epp{M}{R''}=\epp{N}{R''}$ and $\epp{M'}{R''}=\epp{N'}{R''}$. By induction, $\llbracket N\rrbracket\rightarrow^*_{\llbracket D\rrbracket} \mathcal{N}_N$ and $N'\rightarrow^*_{\llbracket D\rrbracket} N''$ for $\mathcal{N}_N\sqsupseteq\epp{N''}{}$. For any role $\role{R}$ we therefore get $\llbracket N\rrbracket_R\xrightarrow{\mu_0}_{\llbracket D\rrbracket(R)}\xrightarrow{\mu_1}_{\llbracket D\rrbracket(R)}\dots B_R$ for $B_R\sqsupseteq \epp{N''}{R}$ for some sequences of transitions $\xrightarrow{\mu_0}_{\llbracket D\rrbracket(R)}\xrightarrow{\mu_1}_{\llbracket D\rrbracket(R)}\dots$ and from the network semantics we get \[\llbracket M\rrbracket \rightarrow^* \prod_{R\in \roles(\type(N))}R[\llbracket V\rrbracket_R~B_R] \mid \prod_{R'\notin \roles(\type(N))} R'[B_{R'}]=\mathcal{N}\] and \[M' \rightarrow^* V~N''\] and the result follows.

\item Assume $M=M''~N$, $M'=M''~N'$, $N\xrightarrow{\tau,\mathbf{R}} N'$, and $\roles(M)\cap \mathbf{R}=\emptyset$. Then for any $R\in \mathbf{R}$, $\roles(\epp{M''}{R})\cap\mathbf{R}=\emptyset$ and the result follows from induction and using \cref{rule:n-app3}. 

\item Assume $M=\case{N}{x}{N'}{x'}{N''}$, $M'=\case{M''}{x}{N'}{x}{N''}$, and $N\xrightarrow{\tau,\mathbf{R}}_D M''$. Then for any role $\role{R}$ such that $R\in\roles(\type(N))$, we have $\llbracket M\rrbracket_R=\case{\llbracket N\rrbracket_R}{x}{\llbracket N'\rrbracket_R}{x'}{\llbracket N''\rrbracket_R}$ and $\llbracket M'\rrbracket_R=\case{\llbracket M''\rrbracket_R}{x}{\llbracket N'\rrbracket_R}{x'}{\llbracket N''\rrbracket_R}$. For any other role $R'$ such that $\epp N{R'}=\epp{N'}{R'}=\epp{N''}{R'}=\botm$, by induction we get $\epp{M''}{R'}=\botm$, and therefore $\epp{M}{R'}=\epp{M'}{R'}=\botm$. For any other role $R''$ such that $\epp{N}{R''}=\botm$, we get $\epp{M}{R''}=\epp{M'}{R''}=\epp{N'}{R''}\merge \epp{N''}{R''}$. For any other roles $R'''$ such that $\epp{N'}{R'''}=\epp{N''}{R'''}=\botm$, we have $\epp{M}{R'''}=\epp{N}{R'''}$ and $\epp{M'}{R'''}=\epp{M''}{R'''}$. For any other role $R''''$, we have $\llbracket M\rrbracket_{R''''}=(\lambda x:\botmt.\llbracket N'\rrbracket_{R''''}\merge\llbracket N''\rrbracket_{R''''})~\llbracket N\rrbracket_{R''''}$ and $\llbracket M'\rrbracket_{R''''}=(\lambda x.\llbracket N'\rrbracket_{R''''}\merge\llbracket N''\rrbracket_{R''''})~\llbracket M''\rrbracket_{R''''}$ for $x\notin \fv(N')\cup \fv(N'')$. The rest follows by simple induction similar to the second case.

\item Assume $M=\case{N}{x}{N_1}{x'}{N_2}$, $M'=\case{N}{x}{N_1'}{x'}{N_2'}$, $N_1\xrightarrow{\tau,\mathbf{R}}_D N_1'$, $N_1\xrightarrow{\tau,\mathbf{R}}_D N_2$, and $\mathbf{R}\cap \roles(N)=\emptyset$. Then for any role $\role{R}$ such that $R\in\roles(\type(N))$, we have $\llbracket M\rrbracket_R=\case{\llbracket N\rrbracket_R}{x}{\llbracket N'\rrbracket_R}{x'}{\llbracket N''\rrbracket_R}$ %
For any other role $R'$ such that $\epp N{R'}=\epp{N_1}{R'}=\epp{N_2}{R'}=\botm$, by induction we get $\epp{N_1'}{R'}=\epp{N_2'}{R'}=\botm$, and therefore $\epp{M}{R'}=\epp{M'}{R'}=\botm$. For any other role $R''$ such that $\epp{N}{R''}=\botm$, we get $\epp{M}{R''}=\epp{N_1}{R''}\merge \epp{N_2}{R''}$. For any other roles $R'''$ such that $\epp{N_1}{R'''}=\epp{N_2}{R'''}=\botm$, we have $\epp{M}{R'''}=\epp{N}{R'''}$. For any other role $R''''$, we have $\llbracket M\rrbracket_{R''''}=(\lambda x:\botmt.\llbracket N_1\rrbracket_{R''''}\merge\llbracket N_2\rrbracket_{R''''})~\llbracket N\rrbracket_{R''''}$. If $\epp{N_1'}{R}\merge \epp{N_2'}{R}$ is defined for all $R$ then the result follows from induction. Otherwise we have $M_1$ and $M_2$ such that $N_1'\xrightarrow{\tau,\mathbf{R}}_D M_1$ and $N_2'\rightarrow{\tau,\mathbf{R}}_D M_2$ and $\epp{M_1}{R}\merge \epp{M_2}{R}$ for all $R$, and the result follows from induction on these transitions.

\item Assume $M=\case{\Inl~V}{x}{N}{x'}{N'}$ and $M'=N[x:= V]$. Then for any role $R\in\roles(\type(\Inl~V))$, we have $\llbracket M\rrbracket_R=\case{\Inl~\llbracket V\rrbracket_R}{x}{\llbracket N\rrbracket_R}{x'}{\llbracket N'\rrbracket_R}$ and $\llbracket M'\rrbracket_R=\llbracket N[x:=\llbracket V\rrbracket_R]\rrbracket_R$. By \cref{thm:ValUnit}, $\llbracket N[x:=\llbracket V\rrbracket_R]\rrbracket_R=\llbracket N\rrbracket_R[x:=\llbracket V\rrbracket_R]$. For any other role $R'\notin\roles(\type(\Inl~V))$, $\epp{\Inl~V}{R'}=\botm$, and therefore $\epp{M}{R'}=\epp{N}{R'}\merge \epp{N'}{R'}\sqsupset \epp{N}{R'}=\epp{M'}{R'}$. The result follows.

\item Assume $M=\case{\Inr~V}{x}{N}{x'}{N'}$ and $M'=N'[x':= V]$. This case is similar to the previous.

\item Assume $M=\case{N}{x}{N_1}{x'}{N_2}$, $M'=\case{N}{x}{N_1'}{x'}{N_2'}$, $N_1\xrightarrow{\mathbf{R}}_DN_1'$, $N_2\xrightarrow{\mathbf{R}}N_2'$, and $\mathbf{R}\cap \roles(N)=\emptyset$. This case is similar to case four.

\item Assume $M=\com{S}{R} V$ and $M'=V[S:=R]$ and $\fv(V)=\emptyset$. Then if $S\neq R$, $\llbracket M\rrbracket_R=\recv{S}~\botm$, $\llbracket M'\rrbracket_R=\llbracket V[S:=R] \rrbracket_R= \llbracket V \rrbracket_R[S:=R]$ since $\roles(\type(V))=S$, $\llbracket M\rrbracket_S=\send{R}~\llbracket V \rrbracket_S$, $\llbracket M' \rrbracket_S=\botm$, and for any $R'\notin \{S,R\}$, $\llbracket M\rrbracket_{R'}=\llbracket M'\rrbracket_{R'}=\botm$. We therefore get $\llbracket M\rrbracket_R\xrightarrow{\recv{S} \llbracket V \rrbracket_S[S:=R]}_{\llbracket D\rrbracket} \llbracket M'\rrbracket_R$, $\llbracket M\rrbracket_S\xrightarrow{\send{R} \llbracket V \rrbracket_S}_{\llbracket D\rrbracket} \llbracket M'\rrbracket_S$, and $\llbracket M\rrbracket_{R'}=\llbracket M'\rrbracket_{R'}$. We define $\mathcal{N}=\mathcal{N}'=\llbracket M'\rrbracket$ and the result follows. If $S=R$, then $\llbracket M\rrbracket_R=(\lambda x. x)~\llbracket V \rrbracket_R$ and $\llbracket M'\rrbracket_R=\llbracket V \rrbracket_R$ and $\mathcal{N}=\mathcal{N}'=\llbracket M'\rrbracket$ and the result follows.

\item Assume $M=\select{S}{R}~l~M'$. Then $\llbracket M \rrbracket_{S}=\oplus_{R}~l~\llbracket M' \rrbracket_{S}$, $\llbracket M \rrbracket_{R}=\&\{l:\llbracket M' \rrbracket_{R}\}$, and for any $R'\notin \{S,R\}$, $\llbracket M \rrbracket_{R'}=\llbracket M' \rrbracket_{R'}$. We therefore get $\llbracket M \rrbracket\xrightarrow{\tau_{R,S}}_{\llbracket D\rrbracket} \llbracket M \rrbracket\setminus \{R,S\}\mid R[\llbracket M' \rrbracket_{R}]\mid S[\llbracket M' \rrbracket_{S}]$ and the result follows.

\item Assume $M=\select{S}{R}~l~N$, $M'=\select{S}{R}~l~N'$, $N\xrightarrow{\tau,\mathbf{R}}_D N'$, and $\mathbf{R}\cap \{S,R\}=\emptyset$. Then $\llbracket M \rrbracket_{S}=\oplus_{R}~l~\llbracket N \rrbracket_{S}$, $\llbracket M' \rrbracket_{S}=\oplus_{R}~l~\llbracket N' \rrbracket_{S}$, $\llbracket M \rrbracket_{R}=\&\{l:\llbracket N \rrbracket_{R}\}$, $\llbracket M' \rrbracket_{R}=\&\{l:\llbracket N' \rrbracket_{R}\}$, and for any $R'\notin \{S,R\}$, $\llbracket M \rrbracket_{R'}=\llbracket N \rrbracket_{R'}$ and $\llbracket M' \rrbracket_{R'}=\llbracket N' \rrbracket_{R'}$. The result follows from induction and using \cref{rule:n-off2,rule:n-cho2}.

\item Assume $M=\fst~\Pair~V~V'$ and $M'=V$. Then for any role $R\in \roles(\type(\Pair~M'~V'))$, $\llbracket M \rrbracket_R=\fst~\Pair~\llbracket M' \rrbracket_R~\llbracket V' \rrbracket_R$ and for any other role $R'\notin \roles(\type(\Pair~M'~V')$, we have $\llbracket M \rrbracket_{R'}=\botm$ and $\llbracket M' \rrbracket_{R'}=\botm$. We define $\mathcal{N}=\mathcal{N}'=\llbracket M'\rrbracket$ and the result follows.

\item Assume $M=\snd~\Pair~V~V'$ and $M'=V'$. Then the case is similar to the previous.

\item Assume $M=f(\vec{R})$ and $M'=D(f(\vec{R'}))[{\vec{R'}}:={\vec{R}}]$. Then the result follows from the definition of $\llbracket D\rrbracket$.\qedhere

\item Assume there exists $N$ such that $M\rightsquigarrow N$ and $N\xrightarrow{\tau,\mathbf{R}}_D M'$. Then the result follows from induction and \cref{thm:squigepp}.
\end{itemize}
\end{proof}

\subsection{Proof of \Cref{thm:NetToChor}}\label{pro:NetToChor}
\begin{definition}
Given a network $\mathcal{N}=\prod\limits_{R\in\rho} R[B_R]$, we have $\mathcal{N}\setminus \rho'= \prod\limits_{R\in(\rho\setminus \rho')} R[B_R]$
\end{definition}

\begin{lemma}
For any role $\role{R}$ and network $\mathcal{N}$, if $\mathcal{N}\xrightarrow{\tau_{\mathbf{R}}} \mathcal{N}'$ and $R\notin \mathbf{R}$ then $\mathcal{N}(R)=\mathcal{N}'(R)$.
\end{lemma}
\begin{proof}
Straightforward from the network semantics.
\end{proof}

\begin{lemma}
For any set of roles $\mathbf{R}$ and network $\mathcal{N}$, if $\mathcal{N}\xrightarrow{\tau_{\mathbf{R'}}} \mathcal{N}'$ and $\mathbf{R}\cap \mathbf{R'}=\emptyset$ then $\mathcal{N}\setminus \mathbf{R}\xrightarrow{\tau_{\mathbf{R'}}} \mathcal{N}'\setminus \mathbf{R}$.
\end{lemma}
\begin{proof}
Straightforward from the network semantics.
\end{proof}

\begin{lemma}\label{thm:squignet}
If $\epp{M}{R}\rightsquigarrow B$ then there exists $M'$ such that $M\rightsquigarrow M'$ and $B\equiv \epp{M'}{R}$
\end{lemma}
\begin{proof}
Follows from case analysis on $\epp{M}{R}\rightsquigarrow B$ keeping in mind that $\epp{M}{R}$ cannot be $\botm~\botm$.
\end{proof}

\begin{proof}[Proof of \cref{thm:NetToChor}]
If $\epp M{}\rightarrow^*_{\epp D{}}\mathcal{N}$ uses \cref{rule:n-str} then this follows from \cref{thm:squignet}.
Otherwise we prove this by structural induction on $M$.
\begin{itemize}
	\item Assume $M=V$. Then for any role $\role{R}$, $\llbracket M\rrbracket_R=L$, and therefore $\llbracket M\rrbracket\not\xrightarrow{\tau_{\mathbf{R}}}$. 
	\item Assume $M=N_1~N_2$. Then for any role $R\in \roles(\type(N_1))\cup (\roles(N_1)\cap\roles(N_2)$, $\llbracket M\rrbracket_R= \llbracket N_1\rrbracket_R~\llbracket N_2\rrbracket_R$, for any role $R'$ such that $\epp {N_1}{R'}=\epp{N_2}{R'}=\botm$, we have $\epp M{R'}=\botm$. For any other role $R''$ such that $\epp{N_1}{R''}=\botm$, $\epp M{R''}=\epp{N_2}{R''}$. For any other role $R'''$ such that $\epp{N_2}=\botm$, we get $\epp{M}{R'''}=\epp{N_1}{R'''}$. We then have 2 cases.
		\begin{itemize}
		\item Assume $N_2=V$. Then $\llbracket N_2\rrbracket_R=L$ by \cref{thm:Induction}, and for any $R'$ such that $R'\notin \roles(\type(N_2))\subseteq \roles(\type(N_1))$, by \cref{thm:ValUnit}, $\llbracket N_2\rrbracket_{R'}=\botm$ and therefore $\epp{M}{R'}=\epp{N_1}{R'}$, and we have 5 cases.
		\begin{itemize}
			\item Assume $N_1=\lambda x:T.N_3$. Then for any role $R\in \roles(\type(N_1))$, $\llbracket N_1\rrbracket_R=\lambda x:\epp{T}{R}.\llbracket N_3\rrbracket_R$. And for any role $R'\notin \roles(\type(N_1))$, $\llbracket N_1\rrbracket_R=\botm$. We have two cases, using either \cref{rule:n-absapp} or \cref{rule:n-inabs,rule:n-app1}. 
			
			If we use \cref{rule:n-absapp}, then there exists $R''$ such that $\mathbf{R}=R''$ and $R''\in \roles(\type(N_1))$. We then get $\mathcal{N}=\llbracket M \rrbracket\setminus \{R''\}\mid R''[\llbracket N_3\rrbracket_{R''}[x:=\llbracket N_2\rrbracket_{R''}]]$. We say that $M'=N_3[x:=N_2]$ and the result follows from using \cref{rule:n-absapp} in every role in $\roles(\type(N_1))$.
			
			If we use \cref{rule:n-inabs,rule:n-app1} then there exists $R''$ such that $\mathbf{R}=R''$ and $\epp{N_3}{R''}\xrightarrow{\mu} B$ and $(\mathcal{N}=\epp{M}{}\setminus \{R''\})\mid R''[\lambda x.B~\epp{N_2}{R''}]$. By induction, $N_3\rightarrow_D^* N_3'$ and $(\epp{N_3}{}\setminus \{R''\}\mid R''[B]\rightarrow_{\DP} \mathcal{N}''$ such that $\epp{N_3}{}\supseteq \mathcal{N}''$, and we define $M'\lambda x:T.N_3'~N_2$ and $\mathcal{N}'=(\mathcal{N}\setminus \roles(\type(N_1)))\mid \prod\limits_{R\in\roles(\type(N_3))} R[(\lambda x.\mathcal{N}''(R))~\epp{N_2}{R''}]$ and the result follows by using \cref{rule:c-inabs,rule:c-app1,rule:n-inabs,rule:n-app1}. 

			\item Assume $N_1=\com{S}{R}$. Then if $S\neq R$, $\llbracket M\rrbracket_S=\send{R}~\llbracket N_2\rrbracket_S$, $\llbracket M\rrbracket_R=\recv{R}~\botm$, and for $R'\notin \{S,R\}$, $\epp{N_1}{R'}=\botm=\llbracket M\rrbracket_{R'}$, and therefore $\mathbf{R}=S,R$, and if $S=R$ then $\llbracket N_1\rrbracket_R=\lambda x. x$.

			If $\mathbf{R}=S,R$ then $\mathcal{N}=\llbracket M\rrbracket\setminus \{S,R\}\mid S[\botm]\mid R[\llbracket N_2\rrbracket_S[S:=R]]$. Because $\llbracket N_2\rrbracket_R=\botm$ and $\llbracket N_2\rrbracket_S=V$, $N_2=V$. Therefore $M\xrightarrow{\mathbf{R}}_D V[S:=R]$ and the result follows.

			If $\mathbf{R}=R$ then $S=R$, $\mathcal{N}=\llbracket M\rrbracket \setminus \{R\} \mid R[\llbracket N_2\rrbracket_R]$ and the rest is similar to above.

			\item Assume $N_1=\fst$. Then $N_2=\Pair~V~V'$ and for any role $R\in \roles(\type(\Pair~V~V'))$, $\llbracket M \rrbracket_R=\fst~\Pair~\llbracket V \rrbracket_R~\llbracket V' \rrbracket_R$ and for any other role $R'\notin \roles(\type(\Pair~V~V')$, by \cref{thm:ValUnit} we have $\llbracket M \rrbracket_{R'}=\epp{N_1}{R'}=\botm$, and therefore $\epp{M}{R'}\not\rightarrow$.

			If $\mathbf{R}=R\in\roles(\type(\Pair~V~V'))$ then $\mathcal{N}= \llbracket M \rrbracket \setminus \{R\} \mid R[\llbracket V \rrbracket_R]$ and $M\xrightarrow{\mathbf{R}}_D V$. The result follows by use of \cref{rule:n-proj1,,thm:ValUnit}.

			\item Assume $N_1=\snd$. This case is similar to the previous.

			\item Otherwise, $N_1\neq V$ and either $\mathbf{R}=R$ or $\mathbf{R}=R,S$.

			If $\mathbf{R}=R$ then either $\llbracket N_1\rrbracket_R\xrightarrow{\tau} B$ and $R\in\roles(\type(N_1))$, $\mathcal{N}=\llbracket M\rrbracket\setminus \{R\} \mid R[B~\llbracket N_2\rrbracket_R]$. We therefore have $\llbracket N_1\rrbracket\xrightarrow{\tau_R} \llbracket N_1\rrbracket\setminus \{R\}\mid R[B]$, and by induction, $N_1\rightarrow^*_{D} N_1'$ such that $\llbracket N_1\rrbracket\setminus \{R\}\mid R[B]\rightarrow^*\mathcal{N}_1\sqsupseteq \llbracket N_1'\rrbracket$. Since all these transitions can be propagated past $N_2$ in the network and $\llbracket N_2\rrbracket_{R'}$ in any role $R'$ involved, we get the result for $M'=N_1'~N_2$.

			If $\mathbf{R}=R,S$ then the case is similar.
		\end{itemize}
		\item If $N_2\neq V$ then we have 2 cases.
		\begin{itemize}
		\item If $\mathbf{R}=R$ then either $\llbracket N_1\rrbracket_R\xrightarrow{\tau} B$ or $\llbracket N_2\rrbracket_R\xrightarrow{\tau} B$ and the case is similar to the previous.

		\item If $\mathbf{R}=S,R$ then there exists $L$ such that either $\llbracket N_1\rrbracket_S\xrightarrow{\send{R}~L} B_S$ or $\llbracket N_2\rrbracket_S\xrightarrow{\send{R}~L} B_S$ and $\llbracket N_1\rrbracket_R\xrightarrow{\recv{S}~L[S:=R]} B_R$ or $\llbracket N_2\rrbracket_R\xrightarrow{\recv{S}~L[S:=R]} B_R$.

		If $\llbracket N_1\rrbracket_S\xrightarrow{\send{R}~L} B_S$ then $\llbracket N_1\rrbracket_S\neq L'$ and therefore $\llbracket N_1\rrbracket_R\xrightarrow{\recv{S}~L[S:=R]} B_R$ and the case is similar to the previous. If $\llbracket N_2\rrbracket_S\xrightarrow{\send{R}~L} B_S$ then $\llbracket N_1\rrbracket_S= L'$, and therefore $\llbracket N_2\rrbracket_R\xrightarrow{\recv{S}~L[S:=R]} B_R$ and the case is similar to the previous.
		\end{itemize}
	\end{itemize}
	\item Assume $M=\case{N}{x}{N'}{x'}{N''}$. Then for any role $R\in \roles(\type(N))$, $\llbracket M\rrbracket_R=\case{\llbracket N\rrbracket_R}{x}{\llbracket N'\rrbracket_R}{x'}{\llbracket N''\rrbracket_R}$. And for any other role $R'\notin \roles(\type(N))$, $\llbracket M\rrbracket_{R'}=(\lambda x.\llbracket N'\rrbracket_{R'}\merge\llbracket N''\rrbracket_{R'})~\llbracket N\rrbracket_{R'}$. We have three cases.
	\begin{itemize}
		\item Assume $\mathbf{R}=R\in \roles(\type(N))$. Then we have three cases.
		\begin{itemize}
			\item Assume $N=\Inl~V$. Then $\llbracket N\rrbracket_R=\Inl~\llbracket V\rrbracket_R$ and $\mathcal{N}=\llbracket M\rrbracket\setminus \{R\} \mid R[\llbracket N'[x:= \llbracket V\rrbracket_R]\rrbracket_R]$. We define $M'=N'$ and since $\llbracket N'\rrbracket_{R'}\sqsupseteq \llbracket N'\rrbracket_{R'}\merge\llbracket N''\rrbracket_{R'}$ the result follows from using \cref{rule:n-absapp,rule:n-casel}.
			\item Assume $N=\Inr~V$. Then the case is similar to the previous.
			\item Otherwise, we use either \cref{rule:n-case} or \cref{rule:n-case2}.
			If we use \cref{rule:n-case}, we have a transition $\llbracket N \rrbracket_{R}\xrightarrow{\tau} B$ such that \[\mathcal{N}=\llbracket M\rrbracket\setminus \{R\} \mid R[\case{B}{x}{\llbracket N'\rrbracket_R}{x'}{\llbracket N''\rrbracket_R}]\] and the result follows from induction similar to the last application case.
			
			If we use \cref{rule:n-case2} then $\epp{N'}{R}\xrightarrow{\tau}_\DP B$ and $\epp{N''}{R}\xrightarrow{\tau}_\DP B$. If $\epp{N'}{R}\xrightarrow{\tau}_\DP B$ then by induction, $N'\rightarrow^*_{D} N'''$ and $\epp{N'}{}\setminus\{R\} \mid R[B]\rightarrow^*_{\DP} \mathcal{N}''$ such that $\mathcal{N}''\supseteq \epp{N'''}{}$ and $N''\rightarrow^*_{D} N''''$ and $\epp{N''}{}\setminus\{R\} \mid R[B]\rightarrow^*_{\DP} \mathcal{N}'''$ such that $\mathcal{N}'''\supseteq \epp{N''''}{}$. Since $N'$ and $N''$ are mergeable on other roles, the result follows from using \cref{rule:c-incase}. 
		\end{itemize}
		\item Assume $\mathbf{R}=R\notin \roles(\type(N))$. Then we have three cases.
		\begin{itemize}
			\item Assume $N=\Inl~V$. Then $\llbracket N\rrbracket_{R}=\botm$ and $\mathcal{N}=\llbracket M\rrbracket\setminus \{R\} \mid R[\llbracket N'\rrbracket_{R}\merge\llbracket N''\rrbracket_{R}]$. We define $M'=N'$ and the result follows.
			\item Assume $N=\Inr~V$. Then the case is similar to the previous.
			\item Otherwise, $\llbracket N \rrbracket_R\neq L$ and we therefore have $\llbracket N \rrbracket_{R}\xrightarrow{\tau} B$ and $\mathcal{N}=\llbracket M\rrbracket\setminus \{R\} \mid R[(\lambda x.\llbracket N'\rrbracket_{R}\merge\llbracket N''\rrbracket_{R})~B]$. We therefore have $\llbracket N\rrbracket\xrightarrow{\tau_R} \llbracket N\rrbracket\setminus \{R\}\mid R[B]$, and by induction, $N\rightarrow_D N'''$ such that $\llbracket N\rrbracket\setminus \{R\}\mid R[B]\rightarrow^* \mathcal{N}'''$ for $\mathcal{N}'''\sqsupseteq\llbracket N'''\rrbracket$. Since all these transitions can be propagated past $N_2$ in the network and the conditional or $(\lambda x.\llbracket N'\rrbracket_{R''}\merge\llbracket N''\rrbracket_{R''})$ in any other role $R'$ involved, we get the result for $M'=\case{N'''}{x}{N'}{x'}{N''}$.
		\end{itemize}
		\item Assume $\mathbf{R}=S,R$. Then the logic is similar to the third subcases of the previous two cases.
	\end{itemize}
	\item Assume $M=\select{S}{R}~\ell~N$. This is similar to the $N_1=\com{S}{R}$ case above.
	\item Assume $M=f( R_1,\dots, R_n)$. Then $\llbracket M\rrbracket=\prod_{i=1}^n R_i[f_i(R_1,\dots, R_{i-1},R_{i+1},\dots, R_n)] \mid \prod_{R\notin \{ R_1,\dots, R_n\}} R[\botm]$. We therefore have some role $\role{R}$ such that $\mathbf{R}=R$ and $\mathcal{N}=(\llbracket M\rrbracket\setminus R_i) \mid R_i[\llbracket D \rrbracket(f_i(\vec{R'}))[{\vec{R'}}:={ R_1,\dots, R_{i-1},R_{i+1},\dots, R_n}]]$. We then define the required choreography $M'=D({f( R'_1,\dots,R'_n)})[{ R'_1,\dots,R'_n}:={ R_1,\dots, R_n}]$ and network \[\mathcal{N}'=\llbracket M'\rrbracket=\prod\limits_{i=1}^n R_i[\llbracket D\rrbracket(f_i(\vec{R'}))[{\vec{R'}}:={ R_1,\dots, R_{i-1},R_{i+1},\dots, R_n}]] \mid \prod\limits_{R\notin  R_1,\dots, R_n} R[\botm]\] and the result follows.\qedhere
\end{itemize}
\end{proof}
\end{document}